\documentclass[11pt,a4paper,dvipsnames]{article}
\usepackage[margin=1in]{geometry}

\usepackage{verbatim}
\usepackage{url,hyperref}
\usepackage{amsmath,amssymb,amsthm,mathrsfs}
\usepackage{parskip}
\usepackage{esvect}
\usepackage{xcolor}

\newcommand{\ignore}[1]{}

\usepackage[T1]{fontenc}
\usepackage[normalem]{ulem}

\usepackage{mathdots}

\usepackage[utf8]{inputenc}

\usepackage{tikz}
\usetikzlibrary{arrows,shapes}
\usetikzlibrary{decorations}
\usetikzlibrary{shapes.geometric, arrows, arrows.meta, calc, decorations.markings}

\usepackage{fullpage}
\usepackage[margin=1in]{geometry}
\usepackage{bbm}
\usepackage[english]{babel}
\usepackage{yhmath}
\usepackage{csquotes}
\usepackage{lmodern}
\usepackage{bm}
\usepackage{float}

\usepackage{hyperref}
\usepackage[vcentermath]{youngtab}
\usepackage{ytableau}\ytableausetup{centertableaux}
\usepackage{amsmath}
\usepackage{amssymb}
\usepackage{amsfonts}
\usepackage{amsthm}
\usepackage{longtable}
\usepackage{mathtools}
\usepackage{mathrsfs}
\usepackage{physics}
\usepackage{todonotes}
\usepackage{multirow}

\usepackage{comment}
\usepackage{assoccnt}
\usepackage[labelfont=bf]{caption}

\usepackage[nameinlink]{cleveref}

\theoremstyle{plain}
\newtheorem{theorem}{Theorem}[section]

\theoremstyle{plain}

\newtheorem{corollary}[theorem]{Corollary}
\newtheorem{lemma}[theorem]{Lemma}

\theoremstyle{definition}
\newtheorem{definition}[theorem]{Definition}

\newtheorem{claim}[theorem]{Claim}

\newtheorem{examplex}[theorem]{Example}
\newenvironment{example}
  {\pushQED{\qed}\examplex}
  {\popQED\endexamplex}

\newtheorem{remarkx}[theorem]{Remark}
\newenvironment{remark}
  {\pushQED{\qed}\remarkx}
  {\popQED\endremarkx}

\crefname{theorem}{theorem}{theorems}
\Crefname{theorem}{Theorem}{Theorems}
\crefname{lemma}{lemma}{lemmata}
\Crefname{lemma}{Lemma}{Lemmata}
\crefname{proposition}{proposition}{propositions}
\Crefname{proposition}{Proposition}{Propositions}
\crefname{claim}{claim}{claims}
\Crefname{claim}{Claim}{Claims}
\crefname{subclaim}{claim}{claims}
\Crefname{subclaim}{Claim}{Claims}
\crefname{corollary}{corollary}{corollaries}
\Crefname{corollary}{Corollary}{Corollaries}
\crefname{subcorollary}{corollary}{corollaries}
\Crefname{subcorollary}{Corollary}{Corollaries}
\crefname{definition}{definition}{definitions}
\Crefname{definition}{Definition}{Definitions}
\crefname{remark}{remark}{remarks}
\Crefname{remark}{Remark}{Remarks}
\crefname{example}{example}{examples}
\Crefname{example}{Example}{Examples}
\crefname{examplex}{example}{examples}
\Crefname{examplex}{Example}{Examples}
\crefname{question}{question}{questions}
\Crefname{question}{Question}{Questions}
\crefname{assumption}{assumption}{assumptions}
\Crefname{assumption}{Assumption}{Assumptions}

\numberwithin{table}{section}

 \DeclareAssociatedCounters{theorem}{table}%
 \DeclareAssociatedCounters{table}{theorem}%

\crefname{table}{table}{tables}
\Crefname{table}{Table}{Tables}

\usepackage{xcolor}
\hypersetup{colorlinks,
    linkcolor={blue!80!black},
    citecolor={green!50!black},
    urlcolor={red!40!black}
}

\newcommand{\bbC}{\mathbb{C}}
\newcommand{\bbN}{\mathbb{N}}

\newcommand{\bbT}{\mathbb{T}}

\newcommand{\bbZ}{\mathbb{Z}}

\newcommand{\calB}{\mathcal{B}}
\newcommand{\calI}{\mathcal{I}}
\newcommand{\calT}{\mathcal{T}}
\newcommand{\calU}{\mathcal{U}}
\newcommand{\calZ}{\mathcal{Z}}
\newcommand{\calW}{\mathcal{W}}
\newcommand{\frakS}{\mathfrak{S}}
\newcommand{\IN}{\mathbb{N}}%
\newcommand{\IC}{\mathbb{C}}%

\newcommand{\eps}{\epsilon} 
\newcommand{\cc}{\textnormal{cc}} 
\renewcommand{\bar}[1]{\overline{#1}} 
\renewcommand{\hat}[1]{\widehat{#1}}

\newcommand{\GL}{\textnormal{GL}}

\newcommand{\End}{\textnormal{End}}

\newcommand{\id}{\textnormal{id}}
\newcommand{\ad}{\textnormal{ad}}
\newcommand{\ideg}{d} 
\newcommand{\varcount}{k} 
\newcommand{\basissize}{K} 
\newcommand{\projwordsize}{L} 
\newcommand{\boundedtuples}{\IN^{\basissize}_{\leq \projwordsize}} 

\newcommand{\gl}{\mathfrak{gl}}
\newcommand{\diff}{\textnormal{d}}

\newcommand{\frakg}{\mathfrak{g}}
\newcommand{\frakh}{\mathfrak{h}}
\newcommand{\frakn}{\mathfrak{n}}
\newcommand{\frakt}{\mathfrak{t}}

\newcommand{\frakgl}{\mathfrak{gl}}
\newcommand{\calA}{\mathcal{A}}
\newcommand{\la}{\lambda}
\newcommand{\diag}{\textup{diag}}
\newcommand{\linspan}{\textup{span}}

\newcommand{\GZ}{\mathcal{G}\hspace{-0.1em}\mathcal{Z}}

\newcommand{\vvirg}{, \ldots ,}
\newcommand{\ootimes}{\otimes \cdots \otimes}

\newcommand{\poly}{\mathsf{poly}}
\newcommand{\VP}{\mathsf{VP}}
\newcommand{\VBP}{\mathsf{VBP}}
\newcommand{\VF}{\mathsf{VF}}
\newcommand{\VQP}{\mathsf{VQP}}
\newcommand{\metaVP}{\mathsf{metaVP}}
\newcommand{\metaVQP}{\mathsf{metaVQP}}
\newcommand{\VNP}{\mathsf{VNP}}
\renewcommand{\P}{\mathsf{P}}
\newcommand{\NP}{\mathsf{NP}}

\newcommand{\calC}{\mathcal{C}}

\newcommand{\per}{\mathsf{per}}

\newcommand{\mbinom}[2]{\Bigl(\begin{array}{@{}c@{}}#1\\#2\end{array}\Bigr)}

\newcommand{\partinto}[1][]{\mathchoice
  {\scalebox{1.2}[.7]{ $\vdash$}_{#1}\,}
  {\scalebox{1.2}[.7]{ $\vdash$}_{#1}\,}
  {\scalebox{.7}{\scalebox{1.2}[.7]{$\vdash$}}_{#1}\,}
  {\scalebox{.5}{\scalebox{1.2}[.7]{$\vdash$}}_{#1}\,}
}

\newcommand{\hard}{{\textup{hard}}}

\author{\hspace{-1cm}Maxim van den Berg, Pranjal Dutta, Fulvio Gesmundo,  Christian Ikenmeyer,  Vladimir Lysikov}
\title{Algebraic metacomplexity and representation theory}

\usepackage{diagbox}

\newcommand{\ybox}{\picture(1,1)
\put(0,0){\line(0,1){1}\line(1,0){1}}
\put(1,0){\line(0,1){1}}
\put(0,1){\line(1,0){1}}
\endpicture}
\newcommand{\ylabel}[1]{\makebox(1,1){#1}}

\date{}
\begin{document}
\raggedbottom

\maketitle

\begin{abstract}
In the algebraic metacomplexity framework we prove that
the decomposition of metapolynomials into their isotypic components can be implemented efficiently, namely with only a quasipolynomial blowup in the circuit size.
We use this to resolve an open question posed by Grochow, Kumar, Saks \& Saraf (2017).
Our result means that many existing algebraic complexity lower bound proofs can be efficiently converted into isotypic lower bound proofs via highest weight metapolynomials, a notion studied in geometric complexity theory.
In the context of algebraic natural proofs, it means that without loss of generality algebraic natural proofs can be assumed to be isotypic.
Our proof is built on the Poincar\'e--Birkhoff--Witt theorem for Lie algebras and on Gelfand--Tsetlin theory, for which we give the necessary comprehensive background.
\end{abstract}

\bigskip

{\small
MvdB: University of Amsterdam and Ruhr University Bochum; Email:~\texttt{maxim.vandenberg@rub.de}

PD:  National University of Singapore; Email:~\texttt{pranjal@nus.edu.sg}

FG: Institut de Mathématiques de Toulouse; UMR5219 -- Université de Toulouse; CNRS -- UPS, F-31062 Toulouse Cedex 9, France; Email:~\texttt{fgesmund@math.univ-toulouse.fr}

CI: University of Warwick; Email:~\texttt{christian.ikenmeyer@warwick.ac.uk}

VL: Ruhr University Bochum; Email:~\texttt{vladimir.lysikov@rub.de}
}

\bigskip


\bigskip

{\small\textbf{2020 Math.~Subj.~Class.:}
68Q15, 20C35, 16S30
}

\bigskip

{\small\textbf{ACM Computing Classification System:}
Theory of computation $\to$ Computational complexity and cryptography $\to$ Algebraic complexity theory
} 

\bigskip

{\small\textbf{Acknowledgments:} We are grateful to Alessandro Danelon and Nutan Limaye for discussions.
We thank anonymous referees for comments, and especially for bringing our attention to Open Question~2 in~\cite{grochow2017towards}.

CI was supported by EPSRC grant EP/W014882/1. PD is supported by the project titled~``Computational Hardness of Lattice Problems and Implications".
MvdB and VL acknowledge support by the European Research Council through an ERC Starting Grant (grant agreement No.~101040907, SYMOPTIC). Views and opinions expressed are however those of the author(s) only and do not necessarily reflect those of the European Union or the European Research Council Executive Agency. Neither the European Union nor the granting authority can be held responsible for them.
}

\thispagestyle{empty}
\setcounter{page}{0}
\newpage

\section{Introduction}\label{sec: intro}

Strong computational complexity lower bounds are often difficult to prove, for many different complexity measures. In most cases, the inability to prove lower bounds can be explained by barrier results such as the \emph{Relativization barrier} due to Baker, Gill and Solovay~\cite{baker1975relativizations}, the \emph{Algebrization barrier} due to Aaronson and Wigderson~\cite{aaronson2009algebrization} and the  \emph{Natural Proofs barrier} due to Razborov and Rudich~\cite{razborov1994natural}. Metacomplexity studies the complexity of computing various complexity measures, such as the minimum circuit complexity of a function.
One eventual goal of the metacomplexity program is to understand mathematically why lower bounds are hard to prove.
In this paper, we study metacomplexity in the algebraic circuit complexity setup from a representation-theoretic viewpoint.

Algebraic complexity focuses on proving complexity lower bounds for various problems related to the computation of multivariate polynomials and rational functions. Prominent examples include the discrete Fourier transform, matrix multiplication, and solving systems of linear equations; see~\cite{burgisser2013algebraic}.
In algebraic complexity, the main model of computation is the {\em algebraic circuit}, which is a labeled directed graph encoding an algorithm to evaluate a polynomial. Most polynomials require large algebraic circuits; however, exhibiting explicit ones with such property is difficult. Valiant's flagship conjecture~\cite{valiant1979completeness}, known as $\VP \neq \VNP$, states that the algebraic circuit complexity of the permanent polynomial $\per_d(x_{1,1},\dots,x_{d,d}) = \sum_{\pi\in\mathfrak{S}_d}\prod_{i=1}^d x_{i,\pi(i)}$ is not polynomially bounded as a function of $d$. This conjecture is an algebraic analogue of the $\P \neq \NP$ conjecture: indeed, the permanent of the adjacency matrix of a bipartite graph counts the number of perfect matchings, which is a $\# \P$-hard problem.

Many current proof techniques in algebraic complexity theory are implicitly based on {\em metapolynomials}. Examples of such techniques are \emph{rank methods} \cite{LanOtt:NewLowerBoundsBorderRankMatMult,DBLP:conf/focs/Limaye0T21}, methods based on geometric features \cite{LaMaRa:DegDuals,GGIL:DegreeRestrictedStrengthDecompsABP}, and the geometric complexity program \cite{MS01,mulmuley2012gct}. We refer to \Cref{subsec:lb-isotopic-hwv} for a more extensive discussion. Informally, metapolynomials are polynomials whose variables are coefficients of polynomials. The concept dates back to resultants and invariants of forms \cite{Cay:TheoryLinTransformations,Sylv:PrinciplesCalculusForms}, while the term ``metapolynomial'' was coined in \cite{Reg:02} and made popular in \cite{grochow2017towards}. A classical example is the metapolynomial $b^2 - 4ac$ in the coefficients of a univariate degree $2$ polynomial $ax^2 + bx +c$, which vanishes exactly when the two roots of the polynomial coincide.
Grochow pointed out in \cite{grochow2015unifying} that for many border complexity measures the complexity lower bounds can without loss of generality be proved by metapolynomials that lie in separating modules, or be isotypic, or be highest weight vectors.
These are all closely related concepts (and we prove analogous results for all of them), where being isotypic is the most natural one, because it does not depend on choosing any bases.
Being a \emph{homogeneous} metapolynomial is the same as being isotypic with respect to the action of $1\times 1$ matrices.
Being isotypic also generalizes the property of being \emph{invariant} under the special linear group.
While invariants can have large circuit size, see \cite{GIMOWW20},
we prove that not all separating isotypic metapolynomials have large circuit size, unless \emph{all} separating metapolynomials have large circuit size.
More precisely, we prove that w.l.o.g.\ one can assume the metapolynomial to be isotypic, or to be a highest weight vector, or to lie in a separating module, with only a quasipolynomial blowup in the circuit size.
We use this to answer Open Question~2 in \cite{grochow2017towards} by proving that \emph{all} metapolynomials that lie in an irreducible representation have the same circuit complexity up to quasipolynomial blowup, see \Cref{cor:openquestiontwo}.

Our result reduces the search space for lower bounds from \emph{all} homogeneous metapolynomials to the isotypic metapolynomials, which is a technique for example used in 
\cite{AbCh:BrillGordanLoci,BI:11,BI:13,HIL13,AIR:16,CIM:17,BHIM:22}. In particular the large-scale cluster computations in \cite{AIR:16,CIM:17,BHIM:22} would be infeasible without this search space reduction. \cite{HIL13} reduces the search space for proving that the border complexity of the $2\times 2$ matrix multiplication tensor is 7 to a 4-dimensional linear subspace of isotypic metapolynomials in a $\binom{64+20-1}{20} = 8,179,808,679,272,664,720$-dimensional space of homogeneous degree 20 metapolynomials on $4\times 4 \times 4$ tensors. On the reduced search space, the problem is then easily solved by linear algebra.
This specific approach was used in \cite{BCZ18} to prove that the border support rank of the $2 \times 2$ matrix multiplication tensor is~7. To this date, this is the only proof of this fact. 
Similar methods are used in \cite{OedSam:5thSecant5P1s,GHL:LinearPreservers}, where the polynomials of interest are invariants for the action of a special linear group: this information is crucially used to reduce the search space and determine the polynomial via interpolation methods on a small space.
These are finite results, but the techniques work in the asymptotic settings as well, for example for lower bounds on the rank of matrix multiplication as in \cite{BI:11,BI:13}, or the
best known lower bound for the border determinantal complexity of the permanent
\cite{LaMaRa:DegDuals}.

Our result also points to potential future lower bounds:
Our method can be used to extract isotypic metapolynomials for example from rank-based lower bounds proofs, and the obtained metapolynomials give the same lower bounds.
Generally, metapolynomials are \emph{not} affected by barriers for rank methods \cite{EGOW:Barrier,Gal:VBcactus}, so the metapolynomials obtained in this way can potentially be adjusted to prove better lower bounds.

Our procedure uses explicit efficient circuit constructions and hence makes a contribution to algorithmic representation theory that is of independent interest.

\paragraph{Metapolynomials} We work over the field $\IC$ of complex numbers, and the set $\{x_1,x_2,\ldots\}$ of variables.
For every sequence $\bm{i}$ of nonnegative integers with $\sum_j i_j < \infty$,
a monomial~$x^{\bm{i}}$ is the finite product of variables $x_1^{i_1}x_2^{i_2}\cdots$, and a polynomial $f$ is a finite linear combination of monomials.
For every monomial $x^{\bm{i}} = x_1^{i_1}x_2^{i_2}\cdots$, define a \emph{metavariable} $c_{\bm{i}}$.
A metavariable $c_{{\bm i}}$ has degree 1, and it has \emph{weight}~$\bm{i}$.
Finite products of metavariables are called metamonomials, and finite linear combinations of metamonomials are called metapolynomials. 
The weight of a metamonomial is the sum of the weights of its metavariables.
If every metamonomial in a metapolynomial $\Delta$ has the same weight, then $\Delta$ is called a weight metapolynomial (and the zero metapolynomial has every weight).
For a fixed weight $w$ the metapolynomials of that weight form a vector space, the weight $w$ space.
Every metapolynomial can be written as a unique sum of weight metapolynomials of different weights.
A metapolynomial is called affine linear if all its metamonomials have degree at most~1, i.e., are just a metavariable or a constant.
Every metapolynomial $\Delta$ can be evaluated at any polynomial $f$ by assigning to each metavariable the coefficient of the corresponding monomial in~$f$.
For example, the discriminant of a degree two homogeneous polynomial in two variables $c_{20}x_0^2 + c_{11}x_0x_1 + c_{02}x_1^2$ is the metapolynomial $\Delta = c_{11}^2-4c_{20}c_{02}$. It is a classical fact that $\Delta(f)=0$ if and only if $f = \ell^2$ for some homogeneous linear polynomial $\ell = \alpha_0 x_0 + \alpha_1 x_1$.
In principle, metapolynomials can involve metavariables that correspond to monomials of different degrees, but for our purposes it is sufficient to consider the case where all metavariables correspond to monomials of the same degree~$d$, see \Cref{rem:padding}.
Every such metapolynomial can be decomposed uniquely as a sum of its homogeneous degree $\delta$ components,
and we call a homogeneous degree $\delta$ metapolynomial a metapolynomial of format $(\delta,d,k)$ if all its metavariables only involve the variables $x_1,\ldots,x_k$.

\paragraph{Algebraic circuits}
In this paper, we are mostly interested in the algebraic circuit size of metapolynomials. Algebraic circuits for metapolynomials are defined in the same way as for usual polynomials, as follows.

An algebraic circuit is a directed acyclic graph in which each vertex
has indegree~$0$ or~$2$. Every indegree~0 vertex is called \emph{input gate} and it is labelled by an affine linear metapolynomial; every indegree~2 vertex is called a \emph{computation gate} and it is labeled by either ``$+$'' or ``$\times$''; every edge is labeled by a complex number, which is assumed to be 1 if this scalar is omitted; there is exactly one vertex of outdegree~0. An algebraic circuit computes a metapolynomial at every vertex by induction over the directed acyclic graph structure (the labels on the edges rescale the computed values), i.e.,
if the values computed at the children of a gate $v$ are $\Delta_1$ and $\Delta_2$ and the incoming edges to $v$ are labelled with $\alpha_1$ and $\alpha_2$, respectively, then the value computed at $v$ is $\alpha_1\Delta_1+\alpha_2\Delta_2$ or $\alpha_1\alpha_2\Delta_1\Delta_2$, depending on whether $v$ is labelled with a ``$+$'' or a ``$\times$''.
An algebraic circuit \emph{computes} the metapolynomial computed at its outdegree 0 vertex.
The {\em size} of a circuit is the total number of its vertices. The algebraic circuit complexity of a metapolynomial is the {\em minimum} size of an algebraic circuit computing it.
For example, \Cref{fig:cktdisc} shows that the algebraic circuit complexity of the discriminant is at most~6. Algebraic circuits for polynomials are defined analogously.
\begin{figure}
\centering
\begin{tikzpicture}[vertex/.style={draw,circle,regular polygon,regular polygon sides=4,inner sep=0cm, minimum width=1cm}, edge/.style = {draw,-Triangle}]
\node[vertex] (b) at (-2,2.25) {$c_{11}$};
\node[vertex] (a) at (2,2.25) {$c_{20}$};
\node[vertex] (c) at (3,2.25) {$c_{02}$};
\node[vertex] (ac) at (2.5,1) {$\times$};
\node[vertex] (bb) at (-2,1) {$\times$};
\node[vertex] (sum) at (0.25,0) {$+$};
\path[edge] (b) to[bend left] (bb);
\path[edge] (b) to[bend right] (bb);
\path[edge] (a) to (ac);
\path[edge] (c) to (ac);
\path[edge] (ac) to node[midway,above]{$-4$} ($(sum.east)!0.5!(sum.north east)$);
\path[edge] (bb) to ($(sum.west)!0.5!(sum.north west)$);
\end{tikzpicture}
\caption{An algebraic circuit computing the discriminant $\Delta=c_{11}^2-4c_{20}c_{02}$.}
\label{fig:cktdisc}
\end{figure}
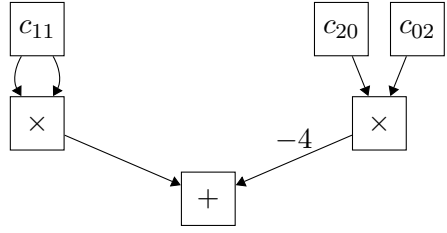

\paragraph{Representation Theory}
A \emph{partition} $\la=(\la_1,\la_2,\ldots)$ is a finite sequence of non-increasing positive natural numbers.
Let $\ell(\la)=\max\{j\mid \la_j\neq 0\}$. We define $\la_j=0$ for all $j>\ell(\la)$.
Write $|\la| := \sum_j \la_j$, and say that $\la$ is a partition of $|\la|$ of length $\ell(\lambda)$. In this case we write $\la \partinto |\la|$ or $\lambda \partinto[\ell(\lambda)] |\lambda|$ depending on whether the length is relevant. For example, $\la=(4,2,2)$ is a partition of $8$ of length $3$, i.e., $\la \partinto[3] 8$.
The \emph{Young diagram} of a partition $\la$ is a top-left justified array of boxes with $\la_j$ many boxes in row~$j$. For example, the Young diagram for $(4,2,2)$ is {\tiny$\yng(4,2,2)$}. A \emph{Young tableau} $T$ is a filling of the boxes of a Young diagram with numbers. The partition $\la$ is called the \emph{shape} of~$T$. For example, {\tiny$\young(1421,31,24)$} is a Young tableau of shape $(4,2,2)$. A Young tableau is \emph{semistandard} if the numbers within each row from left to right are non-decreasing, and the numbers within each column from top to bottom are strictly increasing, for example {\tiny$\young(1122,23,44)$}.
For each partition $\la$ there exists a unique \emph{superstandard} tableau, that is, the tableau with only the number $j$ in row $j$, e.g., {\tiny$\young(1111,22,33)$}.

The group $\GL_k$ acts on the space $\IC[x_1,\ldots,x_k]_d$ of homogeneous polynomials of degree $d$ on $\bbC^k$ by linear change of coordinates. This action induces, again by linear change of coordinates, an action on the space of metapolynomials with format $(\delta,d,k)$ for every $\delta$.
For instance, the change of basis swapping the two coordinates on $\bbC^2$, defines the change of basis on $\bbC[x_0,x_1]_2$ obtained by exchanging $x_0$ and $x_1$. This induces a linear change of coordinates on $\bbC[c_{20},c_{11},c_{02}]$ which exchanges $c_{20}$ and $c_{02}$ and leaves $c_{11}$ fixed. The discriminant $c_{11}^2 - 4c_{02}c_{20}$ is sent to itself by this change of coordinates.
More precisely the group actions are defined as follows. For every homogeneous degree $d$ polynomial $f$ in $k$ variables and every $g \in \GL_k$, the polynomial $g\cdot f$ is defined by $(g \cdot f)(x) = f(g^{-1} x)$ for every $x \in \IC^k$. Similarly, for every format $(\delta,d,k)$ metapolynomial $\Delta$, the metapolynomial $g \cdot \Delta$ is defined by $(g \cdot \Delta)(f)=\Delta(g^{-1} \cdot f)$ for every homogeneous degree $d$ polynomial~$f$. The details about this action are provided in \Cref{sec:lie algebras}.

The vector space of metapolynomials of format $(\delta,d,k)$ is closed under this action of $\GL_k$, and it decomposes into a direct sum of subspaces which are also closed under the action of $\GL_k$: these subspaces are called \emph{isotypic components} and they are in one to one correspondence with partitions $\lambda \partinto[k] {\delta d}$.
The summand corresponding to the partition $\lambda$ is called the $\la$-isotypic component. This is called the isotypic decomposition. Each $\la$-isotypic component decomposes even further into a direct sum of subspaces with one summand for each semistandard tableau~$T$ of shape~$\la$.
We call the component for $T$ the \emph{$T$-isotypic} component.
If $V$ is $\la$-isotypic and \emph{irreducible} (i.e., has no nontrivial subrepresentations), then $\la$ is called the \emph{isomorphism type} of $V$, and each $T$-isotypic component of $V$ is 1-dimensional.
If $T$ is a superstandard tableau of shape $\la$, then the corresponding $T$-isotypic component is called the \emph{highest weight metapolynomial vector space of weight $\la$}.
Equivalently, a highest weight vector of weight $\la$ is a metapolynomial $\Delta$ with $g \cdot \Delta =t_1^{\la_1}\cdots t_k^{\la_k}\Delta$, for every $g \in \GL_k$ upper triangular with $t_1 \vvirg t_k$ on the main diagonal. The $\la$-isotypic component is known to be the linear span of the union of the orbits of all its highest weight vectors.

For example (not easily verifiable by hand, see \Cref{example:projector construction}), the space of metapolynomials of format $(\delta,d,k) = (3,2,3)$ decomposes into three nonzero isotypic components, corresponding to partitions $(6,0,0)$, $(4,2,0)$ and $(2,2,2)$ of $\delta \cdot d = 6$. The isotypic component $(4,2,0)$ further decomposes into three $T$-isotypic components for the three semistandard tableaux of shape $(4,2,0)$: {\tiny$\young(1122,33)$}, {\tiny$\young(1123,23)$}, {\tiny$\young(1133,22)$}. The metapolynomial $\Delta = c_{2, 0, 0}c_{0, 2, 0}c_{0, 0, 2}$ can be written according to this decomposition as
\begin{eqnarray*}
60 \, \Delta
&=&
2\big(c_{1, 0, 1}^2c_{0, 2, 0} + 2c_{1, 1, 0}c_{1, 0, 1}c_{0, 1, 1} + c_{2, 0, 0}c_{0, 1, 1}^2 + c_{1, 1, 0}^2c_{0, 0, 2} + 2c_{2, 0, 0}c_{0, 2, 0}c_{0, 0, 2}\big)
\\
&+&
\Big[\ 2\big(-c_{0,2,0}c_{1,0,1}^2 - 2c_{0,1,1}c_{1,0,1}c_{1,1,0} + 4c_{0,0,2}c_{1,1,0}^2 - c_{0,1,1}^2c_{2,0,0} + 8c_{0,0,2}c_{0,2,0}c_{2,0,0}\big)
\\
&+& \phantom{\Big[\ }0
\\
&+& \phantom{\Big[\ } 5\big( c_{0,2,0}c_{1,0,1}^2 - c_{0,1,1}c_{1,0,1}c_{1,1,0} - c_{0,0,2}c_{1,1,0}^2 + c_{0,1,1}^2c_{2,0,0} + 4 c_{0,0,2}c_{0,2,0}c_{2,0,0}\big)
\ \Big]
\\
&+& 5\big(-c_{1, 0, 1}^2c_{0, 2, 0} + c_{1, 1, 0}c_{1, 0, 1}c_{0, 1, 1} - c_{2, 0, 0}c_{0, 1, 1}^2 - c_{1, 1, 0}^2c_{0, 0, 2} + 4c_{2, 0, 0}c_{0, 2, 0}c_{0, 0, 2}\big),
\end{eqnarray*}
where the five summands from top to bottom are {\tiny$\young(112233)$}-isotypic,
{\tiny$\young(1122,33)$}-isotypic,
{\tiny$\young(1123,23)$}-isotypic,
{\tiny$\young(1133,22)$}-isotypic,
and
{\tiny$\young(11,22,33)$}-isotypic, respectively.
The fifth summand is a highest weight vector, because {\tiny$\young(11,22,33)$} is superstandard.
All summands have weight (2,2,2), because $\Delta$ has weight (2,2,2).

Let $\IC[x_1,\ldots,x_k]_d$ denote the space of homogeneous degree $d$ polynomials in the variables $x_1,\ldots,x_k$.
Let $\IN$ denote the set of natural numbers including zero.

\begin{theorem}[Main theorem]\label{thm:main}
    Let $\Delta \colon \bbC[x_1, \dots, x_k]_{d} \to \bbC$ be a metapolynomial of format $(\delta,d,k)$ computed by an algebraic circuit of size $s$.
    Then
    \begin{enumerate}
        \item For every weight $\mu\in\IN^k$, $|\mu|=d\delta$, the projection of $\Delta$ onto the weight space of weight $\mu$ can be computed by a circuit of size $O(s(\delta d)^{2k^3})$.
        \item For every partition $\lambda \partinto d\delta$ the projection of $\Delta$ onto the $\la$-isotypic component can be computed by a circuit of size $O(sk^{2k^2}(\delta d)^{2k^3})$.
        \item For every partition $\lambda \partinto d\delta$ the projection of $\Delta$ onto the highest weight space of weight $\lambda$ can be computed by a circuit of size $O(s(k+1)^{2k^2}(\delta d)^{2k^3})$.
        \item For every semistandard tableau $T$ of shape $\lambda \partinto d\delta$ the projection of $\Delta$ onto the $T$-isotypic space can be computed by a circuit of size $O(sk^{2k^2}(\delta d)^{2k^4})$.
    \end{enumerate}
\end{theorem}

In all applications of 
\Cref{thm:main},
we will always have that $k$ is logarithmic, which we explain in \Cref{sec: algebraic natural proofs}.
Open Question 2 in \cite{grochow2017towards} asks: given a sequence of irreducible representations of metapolynomials, what is the sequence of metapolynomials in it that has the lowest circuit complexity? \Cref{thm:main} answers this question in a very satisfying way: \emph{Every} such sequence of metapolynomials has the same circuit complexity up to quasipolynomial blowup, as seen in the following corollary.
\begin{corollary}\label{cor:openquestiontwo}
Let $s$ be the smallest circuit complexity of a nonzero metapolynomial in an irreducible $\GL_k$-representation $V$ of type $\la \partinto \delta d$.
Then \emph{every} metapolynomial in $V$ has circuit complexity at most
$O(sk^{2k^2}(\delta d)^{2k^4+k^2})$.
\end{corollary}
\begin{proof}
If $V$ is irreducible of isomorphism type $\lambda$, then it decomposes into a direct sum of  $1$-dimensional $T$-isotypic components, where $T$ ranges over all semistandard tableaux of shape $\lambda$.
Take a nonzero metapolynomial $\Delta$ with minimal complexity in $V$ and apply a generic basis change, i.e., apply a generic element of $\GL_k$.
Then independently project the result onto all $T$-isotypic components.
We claim that since the basis change was generic, the projections are nonzero, and therefore form a basis of $V$.
In fact, we can be more precise about how we pick the basis change element $g\in\GL_k$.
For every semistandard tableau $T$, define $X_T := \{g\in\GL_k \mid g\Delta \text{ has zero $T$-isotypic projection}\}$.
Such $X_T$ is Zariski closed in $\GL_k$, but it is not the whole $\GL_k$,  otherwise the linear span of $\GL_k\Delta$ would be a nontrivial subrepresentation of the irreducible representation $V$.
Therefore the finite union $X:=\bigcup_{T} X_T$ is also a proper Zariski closed subset of $V$, and its complement $\GL_k \setminus X$ is a nonempty Zariski open subset.
If we pick $g \in \GL_k \setminus X$, then the projections of $g\Delta$ to every $T$-isotypic space are nonzero.
Moreover, by \Cref{thm:main}, all these projections have complexity at most $O(sk^{2k^2}(\delta d)^{2k^4})$.
But the dimension of $V$ equals the number of semistandard tableaux of shape $\la$ with entries $1,\ldots,k$, which is at most $(\delta d)^{k^2}$.
\end{proof}

\begin{remark}
\label{rem:padding}
In general, algebraic complexity classes are defined for possibly non-homogeneous polynomials. In this paper, we assume that polynomials and metapolynomials are homogeneous. However, one can lift the results to the nonhomogeneous setting as follows.

A polynomial $f(x_1,\ldots,x_k)$ is called homogeneous if all its monomials have the same degree.
For $d\geq \deg(f)$, define the degree $d$ homogenization
\[f^{\sharp d}(x_0,x_1,\ldots,x_k) := x_{0}^d f(x_1/x_{0},\ldots,x_k/x_{0}).\]
If $f$ has an algebraic circuit of size $s$, then $f^{\sharp d}$ has an algebraic circuit of size at most $O(sd)$:
$s(d+1)$ for extracting the homogeneous parts of $f$ via interpolation,
$d$ operations to compute all values $x_0,x_0^2,\ldots,x_0^d$,
and another $d+1$ for multiplying the homogeneous components with the correct powers of $x_0$, and a final $d$ for adding up the results.
Given a metapolynomial $\Delta$ that does not involve~$x_0$,
define $\Delta^{\sharp d}$ by replacing every metavariable $c_{\bm{i}}$ with $c_{\bm{i}+(d-|\bm{i}|,0,0,\ldots)}$,
and setting every metavariable $c_{\bm{i}}$ with $|\bm{i}|>d$ to zero.
Clearly we have $\Delta^{\sharp d}(f^{\sharp d}) = \Delta(f)$, provided that $d\geq \deg(f)$.
Hence, in this paper we can always assume that $f$ is homogeneous.
\end{remark}

\paragraph{Organization of the paper}
In the first part of \Cref{sec:setting} we illustrate the motivations for studying the circuit complexity of the isotypic components and the highest weight vectors in the space of metapolynomials, in the context of algebraic complexity lower bounds and border complexity classes.
In \Cref{sec: algebraic natural proofs}, based on \Cref{thm:main} we prove that algebraic natural proofs can without loss of generality be assumed to be isotypic, see \Cref{thm:mainalgnat}.
\Cref{sec:outlineofproofofmainthm} sketches the main ideas behind \Cref{thm:main}. In \Cref{sec:lie algebras}, we provide an introduction to Lie algebras, universal enveloping algebras, and their representation theory, on which the core of the proof is built. In \Cref{sec:projector construction}, we explain the construction of the projection maps used in the proof of the main theorem; we describe the efficient circuit implementation in \Cref{subsec:ckttransform}. In \Cref{sec:proofsalgnatproofs}, we prove the results of \Cref{sec: algebraic natural proofs}. In \Cref{sec:pbw}, we show one part of the proof of the Poincar\'e--Birkhoff--Witt Theorem (\Cref{thm:pbw}); this proof is explicit and can be used to construct the projectors in \Cref{thm:main}.

\section{Context and consequences of the main theorem} \label{sec:setting}

\subsection{Complexity classes defined by invariant complexity measures}
A p-family is a sequence of polynomials $(f_n)_{n \in \IN}$ for which the number of variables of $f_n$ and the degree of $f_n$ are polynomially bounded functions of $n$. The class $\VP$ is the set of p-families whose algebraic circuit complexity is polynomially bounded.
Valiant's landmark conjecture is that the permanent sequence $(\per_n)_{n\in\IN}$ with $\per_n = \sum_{\sigma \in \mathfrak{S}_n} \prod_{i=1}^n x_{i,\sigma(i)}$ is not contained in $\VP$, i.e., $\VP\neq\VNP$; for details see~\cite{burgisser2013algebraic,mahajan2014algebraic,burgisser2024completeness}. The extended conjecture states that $(\per_n)\notin\VQP$ \cite{Bur:00}, i.e., that the algebraic circuit complexity of the permanent is not quasipolynomially bounded.

Let $\cc(f)$ denote the algebraic circuit complexity of the polynomial $f \in \bbC[x_1 \vvirg x_k]$. It is easy to see that $\cc(f)=\cc(g \cdot f)$ for any $g\in\GL_k$: A circuit for $g \cdot f$ is obtained from the one for $f$ by applying $g$ at the input gates. Recall that in the definition of algebraic circuits we allow affine linear forms at the input gates. We say that the complexity measure $\cc$ is an \emph{invariant} complexity measure.
In more restrictive circuit definitions, for instance if the input gates are required to be variables, the associated circuit complexity $c'$ might change under the action of $\GL_k$. However, both definitions give the same class $\VP$.

Several other classical classes in algebraic complexity theory can be defined as the set of p-families for which some invariant complexity measure is polynomially bounded.
For example, the class $\VF$ is the set of p-families whose algebraic formula size is polynomially bounded; this is the smallest size of an algebraic circuit whose underlying directed graph is a tree. The class $\VBP$ is the set of p-families whose algebraic branching program width is polynomially bounded; we do not define this notion here, but it is an easy consequence of the definition that it is an invariant complexity measure, e.g. \cite[Def.~2.2]{saptharishi2021survey}. More classical classes can be readily defined in this way, for example the class of p-families of polynomially bounded circuit size and constant depth (still allowing arbitrary linear forms as inputs), such as $\Sigma \Pi \Sigma$, or the class of p-families of sums of polynomially many powering gates $\Sigma\Lambda\Sigma$.
Finally, we mention $\VNP$, which is the class of p-families of polynomially bounded permanental complexity: The smallest $r$ such that $f$ can be written as the permanent of an $r\times r$ matrix with affine linear entries. In all these cases, the relevant complexity measure is invariant.

Notable complexity measures that are \emph{not} invariant are those related to the sparsity of the polynomial, in the sense of \cite{khovanskiui1991fewnomials}, or to restrictions on how often a variable can occur, such as read-once or read-$k$ models \cite{forbes2014polynomial}. A non-invariant complexity measure $c'$ can be converted into an invariant complexity measure $c$ by defining $c(f):=\min\{c'(gf)\mid g \in \GL_k\}$, i.e., by taking the minimum complexity over all base changes, but for some non-invariant measures this can change the corresponding complexity class. In this paper, we only study complexity classes defined by invariant measures.

\subsection{Lower bounds via isotypic metapolynomials} \label{subsec:lb-isotopic-hwv}

An important problem in algebraic complexity theory concerns proving concrete lower bounds for interesting complexity measures~$c$. More precisely, given a polynomial $f_{\textup{hard}}$ and a certain integer $r$, one aims to show $c(f_{\hard}) > r$ or equivalently $f_{\hard} \notin X_r$ where $X_r =\{f \in \IC[x_1,\ldots,x_k]_d \mid c(f) \leq r\}$ is the set of polynomials satisfying the upper bound $r$ on the complexity measure $c(\cdot)$. A standard approach is to determine a function $\Delta$ with the property that $\Delta(f) = 0$ for every $f \in X_r$ and $\Delta(f_\hard) \neq 0$.

Several lower bound methods in complexity theory are based on this approach, and most often the function $\Delta$ is a metapolynomial. This is the case for \emph{rank methods} such as the method of partial derivatives \cite{Sylv:PrinciplesCalculusForms,ChKaWi:PartialDerivativesArithm}, the method of shifted partial derivatives \cite{GKKS:ArithmeticCircuitsChasmDepthThree, forbes2015deterministic}, and other augmented flattening methods such as Koszul--Young flattenings \cite{LanOtt:NewLowerBoundsBorderRankMatMult}: in all these cases the complexity lower bound is obtained in terms of a lower bound of the rank of a matrix whose entries depend on the metavariables, or equivalently on the nonvanishing of a suitable set of minors, which are metapolynomials.  Rank methods yield the recent breakthrough of \cite{DBLP:conf/focs/Limaye0T21} giving the first superpolynomial lower bounds for explicit polynomials to be computed by low-product-depth algebraic circuits in large characteristic, further extended to {\em any} field in \cite{DBLP:conf/coco/000124}. 
 Similarly, the current best known lower bounds for the border determinantal complexity of the permanent, as well as the one of the sum of powers, are based on a non-degeneracy condition which can be translated into the vanishing of a suitable metapolynomial \cite{LaMaRa:DegDuals,DBLP:journals/cc/KumarV22}. Finally, the algebraic branching program lower bounds discussed in \cite{Kum:QuadraticLowerBound,GGIL:DegreeRestrictedStrengthDecompsABP} are built on geometric conditions of the zero set of the polynomial of interest, which in turn can be translated into the vanishing of metapolynomials.
 
When the complexity measure $c(\cdot )$ is invariant, one can use representation theory to enhance the search for the metapolynomials yielding a separation. Our main result \Cref{thm:main} shows that this enhancement is not expensive: the computational overhead is quasi-polynomial. 

To illustrate the effect of our result, first consider a simpler reduction. Let $\Delta$ be a metapolynomial vanishing on $X_r$ such that $\Delta(f_\hard) \neq 0$. Write $\Delta = \sum_i \Delta^{(i)}$ for homogeneous metapolynomials $\Delta^{(i)}$ of degree $i$. If $X_r$ is invariant under rescaling, one has $\Delta^{(i)} (X_r) = 0$ for every $i$. On the other hand, there exists at least one $i$ such that $\Delta^{(i)}(f_\hard) \neq 0$. The relevant homogeneous component $\Delta^{(i)}$ can be extracted from $\Delta$ with a simple interpolation argument, and one has $\cc(\Delta^{(i)}) \leq  (\deg(\Delta)+1) \cdot \cc(\Delta)$. In summary, if the lower bound $c(f_\hard) > r$ can be proved via a metapolynomial $\Delta$ satisfying $\cc(\Delta) \leq s$, then the same lower bound can be proved via a \emph{homogeneous} metapolynomial $\Delta^{(i)}$ satisfying $\cc(\Delta^{(i)}) \leq (\deg(\Delta)+1)s$; and even further, the homogeneous metapolynomial $\Delta^{(i)}$ can be realized as the projection of $\Delta$ onto the $i$-th homogeneous component of the space of metapolynomials.

\Cref{thm:main} shows similar results for other projections and in particular for the projection onto a specific isotypic component. Indeed, suppose $\Delta$ is a (homogeneous) metapolynomial vanishing on $X_r$ and such that $\Delta(f_\hard) \neq 0$. Write $\Delta = \sum_{\lambda} \Delta^{(\lambda)}$ as the sum of its isotypic components. If $c(\cdot)$ is an invariant complexity measure, then $X_r$ is invariant under the action of $\GL_k$. In this case, one has $\Delta^{(\lambda)} (X_r) = 0$ for every~$\lambda$. On the other hand, there exists at least one $\lambda$ such that $\Delta^{(\lambda)}(f_\hard) \neq 0$. \Cref{thm:main} provides an upper bound on the circuit complexity of $\Delta^{(\lambda)}$ in terms of the circuit complexity of $\Delta$, showing that if the lower bound $c(f_\hard) > r$ can be proved via a metapolynomial $\Delta$ satisfying $\cc(\Delta) \leq s$, then the same lower bound can be proved via an \emph{isotypic} metapolynomial with circuit complexity controlled by the parameter $s$. \Cref{thm:main} further proves analogous results for more refined projections, such as the $T$-isotypic projections or the projection on the highest weight space. These are discussed in \Cref{subsec:hwvs}.

\subsection{Highest weight vectors and representation theoretic obstructions}
\label{subsec:hwvs}
\Cref{thm:main} can  be used to convert lower bounds proved by $\Delta$ into lower bounds proved by a highest weight vector (HWV for short). To see this, note that if $\Delta$ vanishes on $X_r$ and $\Delta(f_\textup{hard})\neq 0$, then there exists a highest weight vector $\Delta'$ vanishing on $X_r$ and such that $\Delta'(g \cdot f_\textup{hard})\neq 0$ for some $g\in\GL_k$; in fact, this can be achieved with a random $g \in \GL_k$. 
Hence a HWV proves the lower bound $c(g\cdot f_{\hard}) > r$; if $c(\cdot)$ is an invariant complexity measure, the same lower bound holds for $f_{\hard}$. Similarly to before, \Cref{thm:main} provides an upper bound on the circuit complexity $\cc(\Delta')$ in terms of $\cc(\Delta)$, showing only a quasipolynomial blowup in the complexity of the metapolynomial.

A coarser view on HWVs is given via representation theoretic multiplicities, which is a central idea of geometric complexity theory \cite{MS01,MS08,mulmuley2012gct}.
Let $Y := \overline{\GL_k \, f_\textup{hard}}$ denote the orbit closure, and let $X := \overline{X_s}$.
It is easy to see that $f_\textup{hard} \in X$ if and only if $Y\subseteq X$.
The vector space of HWVs of weight $\la$ defines a vector space of polynomial functions on $X$; the dimension of such space is called 
the multiplicity $\mathsf{mult}_{\lambda}\IC[X]$ of $\lambda$ in the coordinate ring $\IC[X]$. Standard facts in representation theory, and in particular Schur's Lemma (see \Cref{lem:schur}), guarantee that if $\mathsf{mult}_{\lambda}\IC[X] > \mathsf{mult}_{\lambda}\IC[Y]$ for some $\lambda$, then $X \not \subset Y$. Therefore, multiplicities give a method to potentially separate two varieties, and obtain a lower bound on $c(f_\hard)$.
In this case, we say $\lambda$ is a {\em multiplicity obstruction}. If additionally
$\mathsf{mult}_{\lambda}\IC[X] > 0 = \mathsf{mult}_{\lambda}\IC[Y]$, then $\lambda$ is called an {\em occurrence obstruction}.
Occurrence obstructions \cite{BI:11,BI:13} and multiplicity obstructions \cite{IK:20} are sometimes sufficient to obtain separations. However, it is known that occurrence obstructions do not work in certain setups \cite{IP:17,GesIkPa:GCTMatrixPowering,burgisser2019no, DIP:20}.

\Cref{thm:main} makes no statement about the viability of the method of multiplicity obstructions. Indeed, one main idea in geometric complexity theory is that information about $\mathsf{mult}_\la\IC[X]$ can be obtained
without ever having to write down any circuit for a HWV.
Information can be gained
by decomposing much easier coordinate rings of orbits, see \cite{BLMW:11,BurIke:FundamentalInvariantsOrbitClosures}. This gives upper bounds on $\mathsf{mult}_\la\IC[X]$ by using classical representation theoretic branching rules. The lower bounds on $\mathsf{mult}_\la\IC[\overline{\GL_k f_\textup{hard}}]$ are more difficult to obtain, see \cite{IK:20} for a recent implementation.

\subsection{Border complexity}\label{sec:bordercomplexity}

Lower bounds achieved via metapolynomials are better described in the setting of \emph{border complexity}. Generally, for a given algebraic complexity measure $c(\cdot)$, one can define a corresponding \emph{border} measure, $\underline{c}(\cdot)$ as follows: 
\[
\underline{c}(f) = \min \left\{ r : 
\begin{array}{l}
\text{there exists a sequence $f_\eps$ such that $f = \lim_{\eps \to 0} f_\eps$} \\
\text{and $c(f_\eps) = r$ for all but finitely many $\eps \in \bbC$}
\end{array}
\right\};
\]
correspondingly, a complexity class $\calC$ defined in terms of the growth of a given complexity measure naturally induces a \emph{border} complexity class $\bar{\calC}$, described by the growth of the corresponding border complexity measure. In fact, $\bar{\calC}$ can be defined as the closure of $\calC$ with respect to a suitable topology \cite{IkSan22}. In the setting of algebraic circuits, the circuit size $\cc(\cdot)$ has a corresponding border circuit measure $\underline{\cc}(\cdot)$ which defines the complexity class $\bar{\VP}$.

A natural question concerns whether $\calC = \bar{\calC}$ and more generally how much larger $\bar{\calC}$ is compared to $\calC$. For example, the border class of sequences admitting (border) polynomial size $\Sigma\Lambda\Sigma$ circuits is contained in $\VBP$, see \cite{Forbes16} and \cite[Theorem 4.2]{BlDoIk21}. In \cite{DutDwiSax22}, the same containment was shown for the border classes of $\Sigma^{[k]} \Pi \Sigma$ circuits (for every constant $k$). In some restricted setting, for instance in the study of matrix multiplication complexity, it is known that a complexity measure and its associated border complexity are equivalent \cite{BiCaLoRo79}. In general, the question is wide open for most complexity classes and, in particular, for $\VP$.

Metapolynomials characterize border complexity classes in the following sense: if a sequence of polynomials $(f_n)$ does not belong to a border complexity class $\bar{\calC}$ then there is a sequence of metapolynomials \emph{witnessing} the non-membership. This is made precise in \Cref{thm: nullstellensatz for VP}.

\subsection{Algebraic complexity and algebraic natural proofs}\label{sec: algebraic natural proofs}

Razborov and Rudich~\cite{razborov1994natural} introduced the notion of natural proofs in Boolean circuit complexity. 

\begin{definition}[Natural property] \label{def:natural-prop}
A subset $P \subset \{f : \{0,1\}^n \rightarrow \{0,1\}\}$ of Boolean functions is said to be a natural property useful against a class $\calC$ of Boolean circuits if the following is true:
\begin{itemize}
    \item ({\bf Usefulness}). Any Boolean function $f: \{0,1\}^n \rightarrow \{0,1\}$ that can be computed by a Boolean circuit in $\calC$ does not have the property $P$.
    \item ({\bf Constructivity}). Given the truth table of a Boolean function $f: \{0,1\}^n \rightarrow \{0,1\}$, whether it has the property $P$ can be decided in time polynomial in the length of the input, i.e., in time $2^{O(n)}$.
    \item ({\bf Largeness}). For all large enough $n$, at least a $2^{-O(n)}$ fraction of all $n$ variate Boolean functions have the property $P$.
\end{itemize}
\end{definition}
A proof that a certain family of Boolean functions cannot be computed by circuits in $\calC$ is said to be a {\em natural lower bound proof} if the proof (perhaps implicitly) proceeds via establishing a natural property useful against $\calC$, and showing that the candidate hard function has this property. 
Razborov and Rudich showed that most of the known Boolean circuit lower bound proofs, such as lower bounds for $\mathsf{AC}^0$ (constant-depth) circuits have the natural property.  
One can set the same framework in the algebraic setting; see~\cite{grochow2015unifying,grochow2017towards,forbes2018succinct,DBLP:conf/focs/Chatterjee0RST20,DBLP:conf/stacs/0001RST22}. The exact details of the definitions are not fully settled yet, see page 19 of \cite{DBLP:conf/focs/Chatterjee0RST20}. 

Following \cite{grochow2017towards}, the \emph{stretched} classes $\metaVP$ and $\metaVQP$ arise as follows.
The class $\metaVP$ is the class of sequences $(\Delta_1,\Delta_2,\ldots)$ of metapolynomials of format $(\delta(n),n,k(n))$ with $k(n)$ polynomially bounded in $n$, and $\delta(n)$ and $\cc(\Delta_n)$ polynomially bounded in $N(n) = \binom{k(n)+n-1}{n}$, which is the number of degree $n$ monomials in $k$ variables. The class $\metaVQP$ is defined analogously, but with $\cc(\Delta_n)$ quasipolynomially bounded in $N(n)$.

We will now argue that $k(n)$ is bounded by a polynomial in $\log(N(n))$.
To see this, choose $\gamma$ such that $k(n) \le n^{\gamma}$ for $n$ large enough.
Take any such large enough $n$ and make a case distinction.
If $k(n)\leq n$, then $N =\binom{k+n-1}{k-1} \ge (1 + \frac{n}{k-1})^{k-1}$ implies $\log N \ge (k-1) \log (1+\frac{n}{k-1}) \ge (k-1)$, since $\frac{n}{k-1} \ge 1$; here we use the fact that $\binom{m}{r} \ge (m/r)^r$.
On the other hand, if $k(n)\geq n+1$, then since $\binom{m}{r}$ is an increasing function on $m$ for a fixed $r$, we get that~$N = \binom{k+n-1}{n} \ge  \binom{2n}{n} \ge 2^n$, for $n \ge 1$.
Therefore, $(\log N)^{\gamma} \ge n^{\gamma} \ge k$.
Hence, the claim is proved in both cases.

Fix an invariant complexity measure $c$ and let $X_{d,r}$ denote the set of homogeneous degree $d$ polynomials $f$ with $c(f) \leq r$.
For a format $(\delta,d,k)$ metapolynomial $\Delta:\IC[x_1,\ldots,x_k]_d\to\IC$ let $\Delta|_{X_{d,r}} : X_{d,r}\to \IC$ denote the restriction of $\Delta$ to $X_{d,r}$.
Let $\mathcal C$ denote the set of p-families with polynomially bounded~$c$.
The \emph{vanishing ideal} of $\calC$ is defined as
\begin{align*}
I(\mathcal C) := \left\{(\Delta_n)_{n\in\IN} \textup{ of format $(\ast,n,\ast)$ } : 
\begin{array}{l}
\text{for every polynomially bounded } r(n), \\
\text{the sequence } \big(\Delta_{1}|_{X_{1,r(1)}},\Delta_{2}|_{X_{2,r(2)}},\ldots\big) \\
\text{eventually vanishes identically}
\end{array}\right\}
\end{align*}

To emphasize how natural this definition is, we provide a corresponding Nullstellensatz.
\begin{theorem}[Nullstellensatz for $\mathcal C$]\label{thm: nullstellensatz for VP}
    Let $(f_n)_{n \in \mathbb{N}}$ be a p-family of homogeneous polynomials such that $(\deg f_n)$ is a strictly increasing sequence. The sequence $(f_n)$ lies in $\overline{\mathcal C}$ if and only if for every $(\Delta_n) \in I(\mathcal C)$ we have $\Delta_{\deg(f_n)} (f_n) = 0$ eventually.
\end{theorem}

The proof of \Cref{thm: nullstellensatz for VP} is given in \Cref{sec:proofsalgnatproofs}.

\begin{definition}[Algebraic natural proof] \label{def:alg-natural-proof}
We say that $\mathcal C$ has \emph{algebraic natural proofs} if $I(\mathcal C) \mathop{\cap} \metaVQP$ contains a sequence that is not eventually zero.
\end{definition}
Here we loosened the standard definition slightly from $\metaVP$ (see \cite{grochow2017towards,DBLP:conf/focs/Chatterjee0RST20}) to $\metaVQP$. The landmark question in algebraic metacomplexity is whether  $\VP$ has algebraic natural proofs.

\Cref{def:alg-natural-proof} should be compared with \Cref{def:natural-prop} in the Boolean setting. If $I(\VP) \mathop{\cap} \metaVQP$ contains an eventually nonzero sequence $\Delta$ of format $(\delta(n), n, k(n))$, then clearly there exists the largest $r(n)$ such that $\Delta_n(X_{n,r(n)}) = \{0\}$, and further $\cc(\Delta_n) \le 2^{\poly(\log N)}$, where $N:=\binom{k(n)+n-1}{n}$, and $k(n) \le \poly(n)$. Therefore,  $\Delta$ certifies the non-membership, a {\em useful natural property}. Furthermore, it trivially satisfies the {\em largeness} criterion, since $\Delta_n$ does not vanish on most polynomials. Finally, the quasipolynomial bound on the algebraic circuit complexity of $\Delta_n$ gives a natural algebraic analogue of the notion of {\em constructivity}.

Let $\textsf{Isotypic}$ denote the set of sequences of isotypic metapolynomials.
The following theorem follows from \Cref{thm:main}. It states in a concise way that only isotypic metapolynomials have to be considered for algebraic natural proofs.
\begin{theorem}[Main theorem about algebraic natural proofs]\label{thm:mainalgnat}
$\mathcal C$ has algebraic natural proofs if and only if $I(\mathcal C) \mathop{\cap} \metaVQP \mathop{\cap} \textup{\textsf{Isotypic}}$ contains a sequence that is not eventually zero.
\end{theorem}
\begin{proof}
One direction is clear, because $I(\mathcal C) \mathop{\cap} \metaVQP \mathop{\cap} \textup{\textsf{Isotypic}} \subseteq I(\mathcal C) \mathop{\cap} \metaVQP$.

For the other direction, we first consider a general property of isotypic components of vanishing ideals.
Let $I(X)_\delta=\{\Delta \textup{ of format $(\delta,d,k)$} \mid \Delta|_X=0\}$.
Let $X$ be closed under the action of $\GL_k$.
If $\Delta \in I(X)_\delta$, then $g\Delta\in I(X)_\delta$, because for all $x\in X$ we have $(g\Delta)(x)= (\Delta)(g^{-1}x)=0$, since $g^{-1}x \in X$.
Moreover, $I(X)_\delta$ is closed under linear combinations, hence $I(X)_\delta$ is a vector space, and in particular closed under taking limits and thus closed under the action of the Lie algebra.
We have $\Delta\in I(X)_\delta$ if and only if for all $\la$ we have $\Delta^{(\la)}\in I(X)_\delta$, because the projection of $\Delta$ to an isotypic component $\Delta^{(\la)}$ involves only finite linear combinations of Lie algebra actions, see \Cref{sec:effconproj}.

Now, let $\mathcal C$ have algebraic natural proofs, witnessed by a sequence $(\Delta_n)_{n\in\IN} \in I(\mathcal C) \mathop{\cap} \metaVQP$ of format $(\delta(n),n,k(n))$ that is not eventually zero.
Now, let $(\lambda^{n})_{n\in\IN}$ be a sequence of partitions, $\lambda^{n} \partinto[k(n)]{n \cdot \delta(n)}$, such that the isotypic component $\Delta_n^{(\la^{n})}$ of $\Delta_n$ is nonzero whenever $\Delta_n$ is nonzero.
By construction, $(\Delta_n^{(\la_n)})_{n\in\IN} \in \textup{\textsf{Isotypic}}$ and $(\Delta_n^{(\la_n)})_{n\in\IN}$ is not eventually zero.
Since each $X_{n,i}$ is invariant under the group action, we have for all~$i$ and~$n$:
Whenever $\Delta_{n}|_{X_{n,i}} = 0$, then also $\Delta_{n}^{(\la^{n})}|_{X_{n,i}} = 0$.
For every polynomially bounded~$r$ we have
$\Delta_{n}|_{X_{n,r(n)}} = 0$ eventually,
hence we have $\Delta_{n}^{(\la^{n})}|_{X_{n,r(n)}} = 0$ eventually, and hence $(\Delta_n^{(\la_n)})_{n\in\IN} \in I(\mathcal C)$.
Let $N(n)=\binom{k(n)+n-1}{n}$.
Let $s(n) = \cc(\Delta_n) \in 2^{\poly(\log(N))}$.
By \Cref{thm:main}
with $k(n)\in\poly(\log(N))$, $\delta(n)\in\poly(N)$, and $d=n$, we
have $\cc(\Delta^{(\la^n)}_n) \in 2^{\poly(\log(N))}$,
which implies $(\Delta_n^{(\la_n)})_{n\in\IN}\in\metaVQP$,
which finishes the proof.
\end{proof}

\section{Outline of the proof of \texorpdfstring{\Cref{thm:main}}{}}
\label{sec:outlineofproofofmainthm}

The proof of \Cref{thm:main} is built on an efficient construction of the image of metapolynomials under projections onto the isotypic spaces. 

Roughly speaking, we will construct the projectors of \Cref{thm:main} as linear combinations of repeated derivations applied to metapolynomials. These derivations live naturally in a vector space called the \emph{Lie algebra} $\gl_k$ of $\GL_k$. They are defined as the linearization of linear change of coordinates by elements of $\GL_k$ applied to polynomials on which a metapolynomial is evaluated. Repeated applications of these operations is formally described using products in an associative algebra $\mathcal{U}(\gl_k) \supseteq \gl_k$ called the \emph{universal enveloping algebra} of $\gl_k$. 
More concretely, a standard basis element in $\gl_k$ corresponds to an analogue of a shifted partial derivative for metapolynomials (see \Cref{claim:metapolynomials}), and an element of $\mathcal{U}(\gl_k)$ captures repeated applications of such shifted partial derivatives.

Lie algebras and their universal enveloping algebras have a long history in mathematics, and their theory is rich and well-developed. We review the theory in \Cref{sec:lie algebras} to the extent required to construct the projectors. The construction is completely explicit. The universal enveloping algebra $\mathcal{U}(\gl_k)$ contains elements $\{C_1,\ldots,C_r\}$ which \emph{separate} isotypic components in the following sense: each $C_i$ acts by scalar multiplication on each isotypic component, and for any two distinct isotypic components there is a $C_i$ for which the corresponding scalars are different; the precise statement is the content of \Cref{theorem:central-characters} and \Cref{thm: HarishChandra}. This property is used to construct the projectors via an analogue of standard multivariate polynomial interpolation, see \Cref{sec:projector construction}. The elements $\{C_1,\ldots,C_r\}$ are called \emph{Casimir operators}, and they are defined as generators of certain commutative subalgebras of $\mathcal{U}(\gl_k)$. For each case of \Cref{thm:main}, there is a suitable algebra for which the corresponding isotypic components are the spaces of interest, see \Cref{table:projection-algebras-generators-and-diversity}. 

The proof that the projectors can be implemented efficiently is built on the following insights. The Poincar\'e--Birkhoff--Witt theorem (\Cref{thm:pbw}) allows one to express each projector as a linear combination in $\mathcal{U}(\gl_k)$ of $m$-fold products of elements of $\gl_k$ with $m$ bounded polynomially in the number of isotypic components. This number is exponential in $k$, which is polylogarithmic in $N$ as observed in \Cref{sec: algebraic natural proofs}. We use this fact to construct circuits of quasipolynomial size for the projectors in \Cref{subsec:ckttransform}. In turn, this will complete the proof of \Cref{thm:main}.

\section{Representations of Lie algebras}
    \label{sec:lie algebras}

In this section, we provide an introduction to Lie algebras, universal enveloping algebras, and their representation theory, following mainly \cite{FulHar:RepTh, Hall:RepTh}. 
Afterwards, we introduce and explain the results needed to construct the projectors in \Cref{thm:main}. These results are typically stated in the language of abstract Lie algebras, and therefore we take time to build up this general theory. In our case of interest, the Lie algebra arises from $\GL_k$ and many parts simplify. 
We will regularly indicate this in the text for clarity and concreteness.  
We establish homogeneous polynomials and metamonomials as a Lie algebra representation, which we will use as running examples.

A Lie algebra $\frakg$ is a vector space endowed with a bilinear \emph{bracket} operation $[-,-]$ which is skew-commutative, i.e., $[X,Y] = -[Y,X]$, and satisfies a relation called the Jacobi identity $[X,[Y,Z]] +[Z,[X,Y]] +[Y,[Z,X]] = 0$; in particular, the bracket operation is usually not associative. Every associative algebra $\calA$ is a Lie algebra with the commutator bracket given by $[a,b] \coloneqq ab - ba$ for every $a,b \in \calA$. In this paper, we are interested in the representation theory of the general linear group $\GL_k$, and we study it via the representation theory of its Lie algebra $\frakgl_k$: this is the space of $k \times k$ matrices endowed with the commutator bracket. 

A Lie algebra representation is a vector space $V$ with an associated linear map $\rho: \frakg \to \End(V)$ such that $\rho([X,Y]_\frakg) = [\rho(X),\rho(Y)]_{\End(V)}$. We equivalently say that $\frakg$ (or its elements) acts on~$V$. For $v \in V$, we often write $X.v \coloneqq \rho(X)v$ as a shorthand. 
A simple example is $V = \bbC^k$, where $\gl_k$ acts by matrix-vector multiplication. This is called the standard representation of $\gl_k$.
In this paper we examine the space of metapolynomials $\bbC[\bbC[x_1 \vvirg x_k]_d]_\delta]$ as a Lie algebra representation of $\gl_k$. The Lie algebra action is explicitly given in in \Cref{claim:metapolynomials}.

Representation theory of Lie algebras is used to study the representation theory of Lie groups because it allows one to ``linearize'' a group action. 
More precisely, every Lie group~$G$ has an associated Lie algebra $\frakg$ and every Lie group representation (a non-linear map) defines naturally a Lie algebra representation (a linear map). 
All Lie algebra representations we will encounter arise from this correspondence.  We will describe the $\frakgl_k$ action on polynomials and metapolynomials explicitly in \Cref{claim:shifted partials} and \Cref{claim:metapolynomials}, which may be taken as the definition. The details of the correspondence between Lie groups and Lie algebras are not crucial for the rest of the paper. We briefly outline them in \Cref{rmk: liegp to lie alg}. 

For a vector space $V$ with basis $v_1,\ldots,v_r$, let $v_1^*,\ldots, v_r^* \in V^*$  be the dual basis of $V^*$.
Explicitly, $v_\ell^*$ is the linear form mapping a vector $v = \sum_{i} c_i v_i \in V$, with $c_i \in \bbC$, to the coefficient $c_\ell$.
The basis dual to the standard basis of $\bbC^k$ is the set of variables $x_1,\ldots,x_k$. 
Let $E_{i,j} \in \gl_k$ be the $k \times k$ matrix having $1$ at the entry $(i,j)$ and zero elsewhere. Regarding $E_{i,j} : \bbC^k \to \bbC^k$ as a linear map, we have $x_\ell \circ E_{i,j} = 0$ if $\ell \neq i$ and $x_\ell \circ E_{i,j} = x_j$ if $\ell = i$; here $\circ$ denotes the composition of functions. Since $\{E_{i,j}\}_{i,j}$ forms a basis for $\gl_k$, the action of $\frakgl_k$ on a representation $V$ is uniquely determined by the action of the $E_{i,j}$'s.

\begin{claim}[Action on polynomials]
    \label{claim:shifted partials}
    The vector space $\bbC[x_1,\ldots,x_k]_d$ of homogeneous degree $d$ polynomials is a representation of $\gl_k$. It is the linearization of the action of $\GL_k$ acting by base change on the variables $x_1,\ldots,x_k$.
For $f \in \bbC[x_1,\ldots,x_k]_d$, the action is defined by the \emph{shifted partial derivative}
\begin{align}
        \label{equation:shifted partials}
        E_{i,j}.f 
        = 
        -\sum_{\ell=1}^k (x_\ell \circ E_{i,j}) \frac{\partial}{\partial x_\ell} f
        = -x_j \frac{\partial}{\partial x_i} f.
    \end{align}
In particular, for a monomial $f = x_{\ell_1} \cdots x_{\ell_d}$, we have 
\[
E_{i,j}.f = \textstyle \sum_{p=1}^d  x_{\ell_1} \cdots x_{\ell_{p-1}} ( E_{i,j}.x_{\ell_p} ) x_{\ell_{p+1}} \cdots x_{\ell_d}.
\]
\end{claim}

For a multiindex $\mu \in \bbN^k, \sum_i \mu_i = d$, let $c_\mu$ be the corresponding variable, that is, the linear map sending a polynomial $f$ to the coefficient of $x^\mu$. Let $e_1,\ldots,e_k \in \IN^k$ be the standard basis vectors: $e_i$ is the vector with $1$ at the $i$-th entry and $0$ elsewhere.

\begin{claim}[Action on metapolynomials]
    \label{claim:metapolynomials}
    The vector space of metapolynomials of format $(\delta,d,k)$, that is $\bbC[\bbC[x_1,\ldots,x_k]_d]_\delta =  \bbC[c_\mu : \mu \in \bbN^k, |\mu| = d]_\delta$ is a representation of $\gl_k$. It is the linearization of the action of $\GL_k$ acting by base change on the variables $x_1,\ldots,x_k$.
For a metapolynomial $\Delta \in \bbC[c_\mu : \mu \in \bbN^k, |\mu| = d]_\delta$, the action is given by
    \begin{align}
        \label{equation:applying-Eij-to-metapolynomials}
        E_{i,j}. \Delta = \sum_{\mu} (E_{i,j}.c_\mu) \frac{\partial}{ \partial c_\mu}\Delta,
    \end{align}
    where the summation is over all multi-indices $\mu \in \IN^k$ with $\sum_i \mu_i = d$. The action on metavariables is given by 
    \begin{equation}
        \label{equation:applying-Eij-to-metavars}
        E_{i,j}.c_\mu = 
        \begin{cases}
            (\mu_i+1)\, c_{\mu + e_i - e_j} &\text{ if } i \neq j, \\
            \mu_i\, c_{\mu} &\text{ otherwise}.
        \end{cases}
    \end{equation}
    In particular, for a metamonomial $\Delta = c_{\mu_1} \cdots c_{\mu_\delta}$, we have 
\[
E_{i,j}.\Delta = \textstyle \sum_{p=1}^\delta  c_{\mu_1} \cdots c_{\mu_{p-1}} (E_{i,j}.c_{\mu_p}) c_{\mu_{p+1}} \cdots c_{\mu_\delta}.
\]
    \end{claim}

The following remark illustrates the correspondence between representations of complex algebraic Lie groups and representations of the corresponding Lie algebras. As mentioned before, it is not essential for the rest of the paper, but the proof of \Cref{claim:shifted partials} and \Cref{claim:metapolynomials} are built on this correspondence.
\begin{remark}\label{rmk: liegp to lie alg}
    A complex algebraic Lie group is a group with the structure of a complex algebraic variety, \cite[Ch. 7]{FulHar:RepTh}. One can show that $\GL_k$ is such a group. Let $V$ be a finite-dimensional vector space. A representation of a complex algebraic Lie group $G$ on $V$ is a group homomorphism $\rho : G \to \GL(V)$ which is a regular map, in the sense that the entries of $\rho(g) \in \GL(V)$ regarded as a matrix in a fixed basis, are polynomial functions on $G$.
    
   As a vector space, the Lie algebra of $G$ is, by definition, the tangent space of $G$ at its identity element $\frakg = T_{\id} G$. The bracket operation is defined via the \emph{adjoint representation} of $G$ on $\frakg$, see \cite[Ch. 8]{FulHar:RepTh} In the case of $\GL_k$ one has $\frakgl_k \simeq T_{\id_k} \GL_k$ is the space of $k \times k$ matrices, and the bracket is the standard commutator.
    
   Every representation $\rho : G \to \GL(V)$ of $G$ defines a representation $\frakg$ via its differential at the identity: $\diff \rho : T_{\id} G \to T_{\id} \GL(V)$; this is a linear map and defines a Lie algebra homomorphism from $T_{\id} G = \frakg$ to $T_{\id} \GL(V) = \End(V) = \frakgl(V)$. 

    Explicitly, the value of $\diff \rho$ at $X \in \frakg$ is the element $\diff \rho(X) \in \frakgl(V) \simeq \End(V)$ defined by  
\begin{equation}\label{equation:lie-algebra-rep-using-curves}
    \big(\diff \rho (X)\big) v 
    = \frac{\diff}{\diff s}\Big|_{s = 0} \Big( \rho\big(\gamma(s)\big) v \Big) \quad \text{ for every $v \in V$};
\end{equation}
here $\gamma : \bbC \to G$ is any smooth curve such that $\gamma(0) = \id_G$ and $ \gamma'(0) = X$. It is easy to see that $\diff \rho(X)$ does not depend on the choice of $\gamma$. 
\end{remark}

\begin{proof}[Proof of \Cref{claim:shifted partials} and \Cref{claim:metapolynomials}]
\label{proof:lie algebra representations}
We derive the Lie algebra representations from their respective group representations, as in \Cref{rmk: liegp to lie alg}.

Every representation $\rho: G \to \GL(V)$ defines a \emph{pullback} representation on the space of polynomials of degree $p$ on $V$, $\bbC[V]_p$. For $g \in G$ and $f \in \bbC[V]_p$, this is given by $g \cdot f = f \circ \rho (g^{-1})$ where $\rho(g^{-1}) : V \to V$ is regarded as a linear map and $\circ$ denotes the composition of functions. 

The standard representation of $\GL_k$ on $\bbC^k$ induces in this way a representation on the space of polynomials $\bbC[ x_1 \vvirg x_k]_d$. Similarly, the action of $\GL( \bbC[ x_1 \vvirg x_k]_d)$ induces a representation on the space of metapolynomials $ \bbC[ \bbC[ x_1 \vvirg x_k]_d]_\delta$; the action of $\GL_k$ on metapolynomials is the composition of the two representations
\[
\GL_k \to \GL( \bbC[ x_1 \vvirg x_k]_d ) \to \GL( \bbC[ \bbC[ x_1 \vvirg x_k]_d]_\delta ).
\]
Now, if $\rho : G \to \GL(V)$ induces the Lie algebra representation $\diff \rho: \frakg \to \GL(V)$, the pullback representation of $G$ on $\bbC[V]_p$ induces the Lie algebra representation described as follows. For $X \in \frakg$, let $\gamma(s)$ be a smooth curve in $G$ such that $\gamma(0)= \id_G$ and $\gamma'(0) = X$. The chain rule of derivatives implies $\frac{\diff}{\diff s}\big|_{s=0} \gamma(s)^{-1} = -X$. We show this in the case where $G$ is a group of matrices: in this case, one can \emph{multiply} elements of $G$ by elements of $\frakg$ and obtain
\[
\frac{\diff}{\diff s}\bigg|_{s=0} \gamma(s)^{-1} = - \bigl[\frac{\diff}{\diff s}\bigg|_{s=0}\gamma(s) \bigr]
 \cdot \bigl[(\gamma(s)^{-2})\big|_{s=0} \bigr] = -\gamma'(0) \cdot \gamma(0)^{-2}= - X \cdot \id^{-2} = -X.
\]
In general, the multiplication of matrices should be replaced by the action of $G$ on $\frakg$.

Let $V^*$ be the dual space of $V$ with basis $z_1 \vvirg z_r$; in particular $\bbC[V]_p$ is the space of polynomials of degree $p$ in $z_1 \vvirg z_r$. By \cref{equation:lie-algebra-rep-using-curves}, for every $F \in \bbC[V]_p$, $X.F$ is the polynomial function described as follows, on every $\zeta \in V$: 
\begin{align*}
X.F (\zeta)  = &\frac{\diff}{\diff s}\Big|_{0} (\gamma(s) \cdot F )(\zeta) = \\
&\frac{\diff}{\diff s}\Big|_{0} ( F \circ \gamma(s)^{-1} )(\zeta) = \qquad (\text{using chain rule})\\
&\sum_{\ell = 1}^r \frac{\partial}{\partial z_\ell} F ( \zeta) \cdot (z_\ell ( \frac{\diff}{\diff s}\Big|_{0} \gamma(s)^{-1} \zeta )) = \\
&\sum_{\ell = 1}^r \frac{\partial}{\partial z_\ell} F (\zeta) \cdot (z_\ell (-X.\zeta)) = - \frac{\partial}{\partial z_i} F(\zeta ) \cdot z_\ell( X.\zeta).
\end{align*}
Specialize the calculation above to the case of polynomials and metapolynomials to obtain \cref{equation:shifted partials} and \cref{equation:applying-Eij-to-metapolynomials}. It remains to show \cref{equation:applying-Eij-to-metavars}, which follows from a direct calculation: regarding $c_\mu$ as a function of a polynomial $f$, we have 
\[
E_{ij}.c_{\mu} (f)  = - c_\mu ( E_{ij}.f ) = c_\mu (x_j \frac{\partial}{\partial x_i}f) =         \begin{cases}
            (\mu_i+1)\, c_{\mu + e_i - e_j}(f) &\text{ if } i \neq j, \\
            \mu_i\, c_{\mu}(f) &\text{ otherwise}.
        \end{cases}
        \qedhere
\]
\end{proof}

Let $V$ be a representation of a Lie algebra $\frakg$. A subspace $V' \subseteq V$ that is itself a $\frakg$-representation is called a subrepresentation. A representation $V$ is irreducible if it contains no subrepresentations except the zero space of $V$ and $V$ itself. A map $\phi : V \to W$ between two representations is called equivariant if it commutes with the action of $\frakg$. Schur's Lemma characterizes equivariant maps between irreducible representations:
\begin{lemma}[Schur's Lemma, {\cite[Lem.~1.7]{FulHar:RepTh}}]\label{lem:schur}
Let $V,W$ be irreducible representations for a Lie algebra $\frakg$ and let $\phi: V \to W$ be an equivariant map. Then either $\phi = 0$ or $\phi$ is an isomorphism. Moreover, if $V=W$, then $\phi = c \cdot \id_V$ for some $c \in \bbC$.
\end{lemma}

A complex algebraic Lie group is called \emph{reductive} if every representation is completely reducible, namely it splits into direct sum of irreducible subrepresentations; the general linear group $\GL_k$ is reductive, see e.g. \cite[Ch.X]{humphreys2012linear}. In the following $\frakg$ is the Lie algebra of a reductive group, and explicitly we are interested in the case where $G = \GL_k$ and $\frakg = \frakgl_k$.

\subsection{Cartan subalgebras and weights}
A subalgebra $\frakt \subseteq \frakg$ is called abelian if $[\frakt , \frakt] = 0$. A Cartan subalgebra $\frakh$ of $\frakg$ is a maximal abelian subalgebra with the property that if $[X,\frakh] \subseteq \frakh$ then $X \in \frakh$. The subalgebra of diagonal matrices in $\frakgl_k$ is a Cartan subalgebra of $\frakgl_k$. It turns out that if $G$ is reductive, then all Cartan subalgebras of $\frakg$ have the same dimension; in fact they are all equivalent under a natural action of $G$ on $\frakg$ \cite[Sec. D3]{FulHar:RepTh}. In the case of $\GL_k$, the action on $\frakgl_k$ is given by conjugation: $g \cdot X = g X g^{-1}$ for $g \in \GL_k$ and $X \in \frakgl_k$. In this case, one can show that any Cartan subalgebra is equivalent to the subspace of diagonal matrices. 

If $\frakg$ is the Lie algebra of a complex reductive Lie group $G$, then Cartan subalgebras of $\frakg$ correspond to maximal abelian subgroups in $G$. For a fixed Cartan subalgebra $\frakh$, let $\bbT$ be the corresponding abelian subgroup of $G$. If $\frakh$ is the subset of diagonal matrices in $\frakgl_k$, then $\frakh$ is the Lie algebra of the abelian group $\bbT$ of diagonal matrices with nonzero diagonal entries. 
Irreducible $\bbT$-representations are $1$-dimensional and they are in one-to-one correspondence with a lattice $\Lambda$ in the space $\frakh^*$, called the \emph{weight lattice} of $\frakg$; the elements of $\Lambda$ are called \emph{weights}. When $\bbT$ is the subgroup of $\GL_k$ of diagonal matrices, then this definition coincides with the one used in \Cref{sec: intro} for metapolynomials; see also \Cref{example:metapolynomials-3-2-3-weights}. The construction of the weight lattice is not straightforward for arbitrary Lie algebras, but it is very explicit in the case where $\frakh$ is the subalgebra of diagonal elements in $\gl_k$: 
\begin{claim}
    The weight lattice of $\frakgl_k$ with respect to the Cartan subalgebra $\frakh$ of diagonal elements is the abelian group of integer linear combinations of the diagonal elements, that is, the group of linear functions 
    \begin{align*}
    \alpha: \frakh & \to \bbC \\
 \left( \begin{smallmatrix} h_1 & & \\ & \ddots & \\ & & h_k \end{smallmatrix}\right) &\mapsto \alpha_1 h_1 + \cdots + \alpha_k h_k.
    \end{align*}
    with $\alpha_1 \vvirg \alpha_k \in \bbZ$. In the following, we usually identify the weight $\alpha$ with the $k$-tuple of its integer coefficients $\alpha = (\alpha_1 \vvirg \alpha_k)$.
\end{claim}
\begin{proof}
 Let $\bbT$ be the subgroup of diagonal matrices in $\GL_k$. Let $(\alpha_1 \vvirg \alpha_k) \in \bbZ^k$ be an integer vector. 
 \begin{align*}
 \rho: \bbT &\to \bbC^\times \simeq \GL_1 \\
 \left( \begin{smallmatrix} t_1 & & \\ & \ddots & \\ & & t_k \end{smallmatrix}\right) &\mapsto t_1^{\alpha_1} \cdots t_k^{\alpha_k}
 \end{align*}
 is a group homomorphism, hence it defines a $1$-dimensional (irreducible) representation of $\bbT$. It is not hard to prove that all irreducible representations of $\bbT$ arise in this way.

 The induced $\frakh$-representation is 
 \begin{align*}
 \diff \rho : \frakh & \to \bbC \simeq \frakgl_1 \\
 \left( \begin{smallmatrix} h_1 & & \\ & \ddots & \\ & & h_k \end{smallmatrix}\right) &\mapsto \alpha_1 h_1 + \cdots + \alpha_k h_k.
\end{align*}
In other words, irreducible $\frakh$-representations induced from $\bbT$-representations are in one-to-one correspondence with functions $\alpha \in \frakh^*$ mapping a diagonal element $H \in \frakh$ to an integer linear combination of its diagonal entries. \end{proof}

Every $\frakg$-representation $V$ restricts to an $\frakh$-representation, hence it decomposes into the direct sum of irreducible $1$-dimensional $\frakh$-representations. Since they are uniquely determined by their weights, we have $V = \bigoplus_{\alpha \in \Lambda} V_\alpha^{\oplus m_\alpha}$; in particular $m_\alpha \neq 0$ only for finitely many $\alpha \in \Lambda$. If $m_\alpha \neq 0$, then $\alpha$ is called a weight of $V$, $m_\alpha$ is called \emph{multiplicity} and the direct summand $V_\alpha^{\oplus m_\alpha}$ is called the \emph{weight subspace} of $V$ of weight~$\alpha$. An element $X \in \frakh$ acts simultaneously on the $m_\alpha$ copies of $V_\alpha$ via scalar multiplication by $\alpha(X) \in \bbC$.  The elements of a weight subspace are called \emph{weight vectors}; a basis of $V$ consisting of weight vectors is called a \emph{weight basis}.

\begin{example}\label{example:polynomials-2-3-weights}
Consider the $\gl_3$-representation $V = \bbC[x_1,x_2,x_3]_2$ with the action described in \Cref{claim:shifted partials}. Let $\frakh \subseteq \frakgl_3$ be the subalgebra of diagonal matrices. 
Consider the six monomials of $V$, which define a basis:
\[
\begin{array}{lll}
x_1^2,\quad & x_2^2,\quad & x_3^2, \\
x_2x_3,\quad & x_1x_3,\quad & x_1x_2. 
\end{array}
\]
One can verify that they form a weight basis of $V$. For instance, given $H =\left( \begin{smallmatrix} t_1  & & \\ & t_2 & \\ & & t_3 \end{smallmatrix}\right) = t_1 E_{11} + t_2 E_{22} + t_3 E_{33} \in \frakh$, we have 
\[
H. (x_1^2) = \textstyle \sum_i t_i E_{ii}.(x_1^2) = - \sum_i t_i \sum_j x_j \displaystyle \frac{\partial}{\partial x_j} (x_1^2) = - (2t_1) x_1^2
\]
showing that $x_1^2$ is a weight vector of weight $(-2,0,0)$. More generally, we have 
{\small
\[
\begin{array}{rlrlrl}
H.(x_1^2)   =\hspace{-0.5em}& \hspace{-0.5em} -2t_1\, x_1^2,  & H.(x_2^2)   =\hspace{-0.5em}&  \hspace{-0.5em} -2t_2\, x_2^2,  & H.(x_3^2) =\hspace{-0.5em}& \hspace{-0.5em} -2t_3\, x_3^2, \\
H. (x_2x_3) =\hspace{-0.5em}& \hspace{-0.5em} -(t_2+t_3)\, x_2x_3, & H. (x_1x_3) =\hspace{-0.5em}& \hspace{-0.5em} -(t_1+t_3)\, x_1x_3, & H.(x_1x_2) =\hspace{-0.5em}&\hspace{-0.5em} -(t_1+t_2) \, x_1x_2. 
\end{array}
\]
}
This shows that $V$ decomposes into the sum of six weight spaces, each spanned by one of the weight vectors above. The corresponding weights are the six elements of $\frakh^*$ mapping $H$ to the eigenvalues of $H$ on each weight vector. As $k$-tuple of integers, they are 
\[
\begin{array}{lll}
(-2,0,0),   \quad & (0,-2,0),    \quad&  (0,0,-2), \\
(0,-1,-1),\quad & (-1,0,-1), \quad& (-1,-1,0). 
\end{array}
\]
More generally, in a space of polynomials $V = \bbC[x_1 \vvirg x_k]_d$, the monomials form a weight basis. The weight of a monomial $x^\mu = x_1^{\mu_1} \cdots x_k^{\mu_k}$, regarded as a vector of integers, is $-\mu$.
\end{example}
\begin{example}
\label{example:metapolynomials-3-2-3-weights}
Consider the $\gl_2$-representation $V = \bbC[\bbC[x_1,x_2]_3]_2$ of metapolynomials of format $(2,3,2)$, with the action described in \Cref{claim:metapolynomials}. One can verify that the metamonomials form a weight basis. For instance, given $H =\left( \begin{smallmatrix} t_1  &  \\ & t_2 \end{smallmatrix}\right) = t_1 E_{11} + t_2 E_{22}  \in \frakh$, we have 
\[
H . (c_{30}^2) = \textstyle \sum_{i= 1}^2 t_i E_{ii}. (c_{30}^2) = \textstyle \sum_{i= 1}^2  t_i \sum_{\mu}E_{ii}.c_\mu \cdot \displaystyle \frac{\partial}{\partial c_{\mu}} (c_{30}^2) = 3 \cdot 2 \cdot t_1\cdot c_{30}^2 = (6t_1) c_{30}^2.
\]
showing that $c_{30}^2$ is a weight vector of weight $(6,0)$. Similarly, one can compute the weight of every metamonomial: 
{\small
\[
\begin{array}{rlrlrl}
H.(c_{21}^2) = \hspace{-0.5em}& \hspace{-0.5em} (4t_1+2t_2) c_{21}^2,& H.(c_{12}^2) = \hspace{-0.5em}&\hspace{-0.5em}(2t_1+4t_2) c_{12}^2 ,& H.(c_{03}^2) = \hspace{-0.5em}&\hspace{-0.5em}(6t_2) c_{03}^2, \\ 
H.(c_{30}c_{21}) = \hspace{-0.5em}& \hspace{-0.5em}(5t_1+t_2) c_{30}c_{21} ,&H.(c_{30}c_{12}) = \hspace{-0.5em}&\hspace{-0.5em} (4t_1+2t_2) c_{30}c_{12}, &H.(c_{30}c_{03}) = \hspace{-0.5em}& \hspace{-0.5em}(3t_1+3t_2) c_{30}c_{03}, \\
H.(c_{21}c_{12}) =\hspace{-0.5em}&\hspace{-0.5em} (3t_1+3t_2) c_{21}c_{12}  ,&H.(c_{21}c_{03}) = \hspace{-0.5em}& \hspace{-0.5em}(2t_1+4t_2) c_{21}c_{03} ,&H.(c_{12}c_{03}) = \hspace{-0.5em}&\hspace{-0.5em} (1t_1+5t_2) c_{21}c_{12} .
\end{array}
\]
}
Therefore $V$ has seven weight spaces of weight $(p,6-p)$ with $p =0 \vvirg 6$. The weight spaces of weight $(4,2),(3,3),(2,4)$ have dimension $2$, the others have dimension $1$.

More generally, the metamonomials give a weight basis for the space of metapolynomials of format $(\delta,d,k)$. A metamonomial $c_{\mu_1} \cdots c_{\mu_\delta}$ has weight $\mu_1 + \cdots + \mu_\delta$,
in accordance with what we defined in the introduction.
\end{example}

\subsection{Roots and highest weights}
An important representation is given by the action of the Lie algebra on itself, called the adjoint representation $\ad: \frakg \to \End(\frakg)$ and defined by $\ad(X) = [X,-]$. 
By definition, the weight space of weight $0$ is $\frakh$ itself; moreover, one can prove that all other weight spaces are $1$-dimensional \cite[Ch. 14]{FulHar:RepTh}. The nonzero weights of the adjoint representation are called the \emph{roots} of the Lie algebra $\frakg$ (with respect to $\frakh$); let $\Phi_\frakg \subseteq \Lambda$ denote the set of roots of $\frakg$. It turns out that $\eta \in \frakh^*$ is a root if and only if $-\eta$ is a root and $\pm \eta$ are the only integer multiples of $\eta$ that are roots. 

We compute the roots for $\gl_k$ with $\frakh$ the algebra of diagonal matrices. Let $H = \diag(t_1,\ldots,t_k) \in \frakh$. Then
\begin{align}
    \ad(H)E_{i,j} = H E_{i,j} - E_{i,j} H = t_i E_{i,j} - t_j E_{i,j} 
    = 
    (t_i - t_j) E_{i,j}.
\end{align}
So $\bbC E_{i,j}$ is the weight space corresponding to the weight $(e_i - e_j) \in \frakh^*$. Since the matrices $E_{i,j}$ span $\gl_k$, we obtain that the set of roots is $\Phi_{\gl_k} = \{e_i - e_j \mid i,j \in [k], i \neq j\}$. 

The roots of the Lie algebra \emph{shift} the weights of the Lie algebra representations: more precisely, let $V$ be a representation of $\frakg$, $v \in V$ a weight vector of weight $\lambda$ and $X \in \frakg_\alpha$ for some root $\alpha$; then $X.v$ is a weight vector of weight $\lambda + \alpha$. 

A choice of positive roots is a subset $\Phi_{\frakg}^+$ of $\Phi_{\frakg}$ with the following two properties: for every root $\alpha$, either $\alpha$ or $-\alpha$ belong to $\Phi_{\frakg}^+$; for every pair of roots $\alpha_1,\alpha_2$ such that $\alpha_1+\alpha_2$ is a root, if $\alpha_1,\alpha_2 \in \Phi_{\frakg}^+$ then $\alpha_1 + \alpha_2 \in \Phi_{\frakg}^+$. For $\frakgl_k$, one such choice is $\Phi_{\gl_k}^+ = \{e_i - e_j \mid i < j\}$.

All choices of positive roots are equivalent under the action of a finite group $\calW_\frakg$, called the Weyl group of $\frakg$. This group acts on $\frakh$ and hence on $\frakh^*$; the action of $\frakh^*$ sends the set of roots to itself. All possible sets of positive roots are obtained one from the other via the action of $\calW_\frakg$. The Weyl group of $\frakgl_k$ is the permutation group $\frakS_k$ on $k$ elements. If $\frakh$ is the Cartan subalgebra of diagonal matrices, then $\calW_\frakg$ acts by permuting the diagonal elements; the induced action on $\frakh^*$ permutes the weights $e_i$: for $\sigma \in \frakS_k$, $\sigma \cdot e_i$ is the weight $e_{\sigma^{-1}(i)}$.

Since for every root $\alpha$, either $\alpha$ or $-\alpha$ is positive, the weight decomposition of $\frakg$ gives $\frakg = \frakh \oplus \bigoplus_{\alpha \in \Phi_\frakg^+} (\frakg_{\alpha} \oplus \frakg_{-\alpha})$. In the case of $\frakgl_k$, this decomposition is $\frakgl_k = \frakh \oplus \bigoplus_{i < j} (\bbC E_{ij} \oplus \bbC E_{ji})$. A \emph{raising operator} (resp.\ \emph{lowering operator}) is an element $X \in \bigoplus_{\alpha \in \Phi_\frakg^+} \frakg_{\alpha} $ (resp.\ $X \in \bigoplus_{\alpha \in \Phi_\frakg^+} \frakg_{-\alpha} $). Let $\frakn = \bigoplus_{\alpha \in \Phi_{\frakg}^+}$ be the space of raising operators; the condition that a root $\alpha_1 + \alpha_2$ is positive whenever $\alpha_1,\alpha_2$ are positive make $\frakn$ into a Lie algebra.

For $\frakgl_k$, the space $\frakn = \linspan\{ E_{i,j} : i < j\}$ of raising operators is the space of strictly upper triangular matrices and the space of lowering operators is the space of strictly lower triangular matrices.

Let $V$ be a representation of $\frakg$. A \emph{highest weight vector} for $V$ is a weight vector $v \in V$ with the property that $\frakn.v = 0$. The corresponding weight is called a \emph{highest weight} for $V$. When $\frakg = \frakgl_k$ and $V$ is a space of metapolynomials, it is not hard to verify that this condition is equivalent to the characterization of highest weight metapolynomials in \Cref{sec: intro}; the characterization in terms of tableaux is explained in details in \Cref{sec:gelfand-tsetlin}. A fundamental result in representation theory is that if $V$ is irreducible, then $V$ has a unique highest weight and this highest weight uniquely determines $V$ up to isomorphism \cite[Prop.\ 14.13]{FulHar:RepTh}. In particular, this result classifies finite-dimensional representations of $\frakg$.

Only a subset of weights occur as highest weights of finite-dimensional irreducible representations. Such vectors are called \emph{dominant weights} and their characterization in general requires the introduction of a suitable inner product structure on $\frakh^*$, which is not straightforward.  Let $\Lambda^+$ denote the subset of the weight lattice consisting of dominant weights. The representation associated to the highest weight $\lambda$ is denoted $V_\lambda$ and it is the unique $\frakg$-representation having highest weight $\lambda$. In the case of $\frakgl_k$, the set of dominant weights consist of weights $\lambda = (\lambda_1 \vvirg \lambda_k)$ such that $\lambda_j \geq \lambda_{j+1}$. If $\lambda_i \geq 0$ for every $i$, we usually identify the weight $\lambda$ with a partition of length $k$.

\begin{example}
\label{example:irreps-of-polynomials-2-3-and-metapolynomials-3-2-3}
We determine the highest weights vectors in the representations $\bbC[x_1,x_2,x_3]_2$ and $\bbC[\bbC[x_1,x_2,x_3]_2]_3$. The space of raising operators of $\frakgl_3$ is $\frakn = \textnormal{span}\{E_{12}, E_{23}, E_{13}\}$ of upper triangular matrices. The highest weight vectors are weight vectors characterized by the condition $\frakn. v = 0$. Since $E_{13} = [E_{12}, E_{23}]$, the condition $\frakn. v= 0$ is equivalent to $E_{12}.v = E_{23}.v = 0$.  

For $\bbC[x_1,x_2,x_3]_2$, one can check the unique highest weight vector is $x_3^2$ and the corresponding highest weight is $(0,0,-2)$. Therefore $\bbC[x_1,x_2,x_3]_2$ is irreducible, and it is isomorphic to the $\frakgl_3$-representation $V_{(0,0,-2)}$. 

For $\bbC[\bbC[x,y,z]_2]_3$, the highest weight vectors are the metapolynomials
\begin{equation*}
\begin{gathered}
    c_{2,0,0}^3,
    \qquad\quad
    c_{2,0,0}c_{1,1,0}^2 - 4 c_{2,0,0}^2c_{0,2,0}
    \qquad\quad\text{and}
    \\
     4 c_{2,0,0}c_{0,2,0}c_{0,0,2} - c_{0,0,2}c_{1,1,0}^2 
    + 
    c_{1,0,1}c_{1,1,0}c_{0,1,1} - c_{0,2,0}c_{1,0,1}^2 - c_{2,0,0}c_{0,1,1}^2.
\end{gathered}
\end{equation*}
Their weights are $(6,0,0)$, $(4,2,0)$ and $(2,2,2)$. Therefore $\bbC[\bbC[x_1,x_2,x_3]_2]_3$ is a direct sum of three irreducible $\frakgl_3$-representations
\[
\bbC[\bbC[x_1,x_2,x_3]_2]_3 = V_{(6,0,0)} \ \oplus \ V_{(4,2,0)}\ \oplus \ V_{(2,2,2)}. \qedhere
\]
\end{example}

\subsection{Universal enveloping algebra}
The universal enveloping algebra of a Lie algebra $\frakg$ is an associative algebra associated to the Lie algebra $\frakg$. It is used to study representation theory of Lie algebras using representation theory of associative algebras. 

A representation of an associative (unital) algebra $\calA$ is a vector space $V$ equipped with a homomorphism of associative algebras $\rho \colon \calA \to \End(V)$. For $v \in V$ and $U \in \calA$, write $U.v = \rho(U)v$, as in the Lie algebra setting. Subrepresentations and equivariant maps are defined for representations of an associative algebra in the same way as for Lie algebra representations. For a subset $S \subseteq \calA$, the subalgebra generated by $S$ is the smallest subalgebra of $\calA$ containing $S$; explicitly, it can be realized as the set of (noncommutative) polynomials in elements of $S$. 

Let $\frakg$ be a Lie algebra. Consider a representation $\rho \colon \frakg \to \End(V)$.
Since $\End(V)$ is an associative algebra, we can consider the subalgebra of $\End(V)$ generated by the linear maps in the image~$\rho(\frakg)$.
Since $\rho$ is a representation, these generators satisfy the relations 
\begin{equation}\label{eq:lie-representation-relations}
\rho(X) \rho(Y) - \rho(Y) \rho(X) = \rho([X, Y]).
\end{equation}

The \emph{universal enveloping algebra} $\mathcal{U}(\frakg)$ is an associative algebra constructed in such a way that the relations~\eqref{eq:lie-representation-relations} are forced for every representation $\rho$ of $\mathcal{U}(\frakg)$.
It is a quotient of the \emph{tensor algebra} $\mathcal{T}(\frakg)$ over $\frakg$. The tensor algebra is spanned by formal products of elements of $\frakg$,
\begin{align}
\mathcal{T}(\frakg) = \bigoplus_{k = 0}^{\infty} \frakg^{\otimes k} = \bbC \oplus \frakg \oplus \frakg^{\otimes 2} \oplus \frakg^{\otimes 3} \oplus \cdots
\end{align}
with multiplication induced by the tensor products $\otimes \colon \frakg^{\otimes p} \times \frakg^{\otimes q} \to \frakg^{\otimes(p + q)}$. In particular, the tensor algebra is graded, with $\frakg^{\otimes p}$ being the component of degree $p$. Consider the ideal $\mathcal{I}(\frakg) \subset \mathcal{T}(\frakg)$ generated by the elements of the form $$X \otimes Y - Y \otimes X - [X, Y]$$ for $X, Y \in \frakg$.
The universal enveloping algebra of $\frakg$ is defined as $\mathcal{U}(\frakg) = \mathcal{T}(\frakg)/\mathcal{I}(\frakg)$.
The multiplication in $\mathcal{U}(\frakg)$ is denoted by the usual multiplication sign rather than $\otimes$. 
The inclusion of $\frakg$ into $\mathcal{T}(\frakg)$ (as a vector space) gives rise to a linear map $\frakg \to \mathcal{U}(\frakg)$, which is also injective.
We identify $\frakg$ with the image of this injection.
The elements of $\mathcal{T}(\frakg)$ and, therefore, $\mathcal{U}(\frakg)$ can be written as noncommutative polynomials in elements of $\frakg$. When $\frakg$ is abelian (for instance, in the case of a Cartan algebra), the ideal $\mathcal{I}(\frakg)$ is indeed homogeneous, since $[X,Y] = 0$; in this case $\calU(\frakg)$ is the symmetric algebra of $\frakg$, isomorphic to the polynomial ring on $\frakg^*$.

By construction of the ideal $\mathcal{I}(\frakg)$ we have $XY - YX = [X, Y]$ in $\mathcal{U}(\frakg)$ for every $X, Y \in \frakg$. This implies that every representation of $\mathcal{U}(\frakg)$ restricts to a Lie algebra representation of the Lie algebra $\frakg$, via the embedding $\frakg \subseteq \mathcal{U}(\frakg)$; conversely every representation $\rho \colon \frakg \to \End(V)$ of $\frakg$ uniquely extends to a representation of $\mathcal{U}(\frakg)$ by setting $\rho(X_1 X_2 \dots X_\ell) = \rho(X_1)\rho(X_2) \dots \rho(X_\ell)$ for all $X_i \in \frakg$.

Note that in general the ideal $\mathcal{I}(\frakg)$ is not homogeneous, so $\mathcal{U}(\frakg)$ does not inherit the grading from the tensor algebra $\mathcal{T}(\frakg)$. However, $\mathcal{U}(\frakg)$ has a \emph{filtration}, in the following sense: define $\mathcal{U}(\frakg)_{\leq \ell} = \{X + \mathcal{I}(\frakg) \mid X \in \bigoplus_{k = 0}^\ell \frakg^{\otimes k} \}$; then $\mathcal{U}(\frakg)_{\leq \ell} \subseteq \mathcal{U}(\frakg)_{\leq (\ell+1)}$ for every $\ell$, and if $X \in \mathcal{U}(\frakg)_{\leq p},Y \in \mathcal{U}(\frakg)_{\leq q}$ then $XY \in \mathcal{U}(\frakg)_{\leq (p+q)}$.

In particular, the elements of $\mathcal{U}(\frakg)_{\leq \ell}$ are represented by noncommutative polynomials of degree at most $\ell$ in elements of $\frakg$.
The \emph{length} of an element $X \in \mathcal{U}(\frakg)$ is the minimal $\ell$ such that $X \in \mathcal{U}(\frakg)_{\leq \ell}$.

\begin{theorem}[Poincar\'e--Birkhoff--Witt {\cite[Theorem 9.10]{Hall:RepTh}}]
\label{thm:pbw} 
Let $X_1, \dots, X_r$ be a basis of~$\frakg$. Then the set of ordered monomials $X_{i_1} X_{i_2} \cdots X_{i_d}$, $i_1 \leq i_2 \leq \cdots \leq i_d$ is a basis of $\mathcal{U}(\frakg)$.
Moreover, the set of ordered monomials of length at most $\ell$ is a basis of the subspace $\mathcal{U}(\frakg)_{\leq \ell}$ of elements with length at most $\ell$.
\end{theorem}
In fact, we will only use that ordered monomials are a set of generators for the spaces $\mathcal{U}(\frakg)_{\leq \ell}$. This is proved explicitly in \Cref{theorem:pbw}.
The Poincar\'e--Birkhoff--Witt theorem (PBW theorem) shows that elements of $\mathcal{U}(\frakg)$ have unique representation as noncommutative polynomials in $X_1, \dots, X_r$ with ordered terms, and the length of the element is given by the longest monomial in this polynomial. 

\begin{example}
\label{example:universal-enveloping-algebra-gl2}
Consider $\frakg = \gl_2$ with a $\gl_2$-representation $V$. 
Then $F \coloneqq E_{2,1}, E \coloneqq E_{1,2}, H_i \coloneqq E_{i,i}$ gives an ordered basis $(F,E,H_1,H_2)$ of $\gl_2$.
Examples of elements of $\mathcal{U}(\gl_2)$ are $H_1EF$ and $1 - \tfrac12 FE + \tfrac{1}{12}F^2E^2$.
Note that these are formal products, not matrix multiplication (general Lie algebras do not have this additional algebra structure).
Taking $v \in V$, the second element acts as 
$(1 - \tfrac12 FE + \tfrac{1}{12}F^2E^2).v 
= v - \tfrac12 F.(E.v) + \tfrac1{12} F.(F.(E.(E.v))) \in V$.
The PBW theorem guarantees that these elements can be written as polynomials with ordered monomials. Indeed, we have 
\begin{align*}
    H_1EF 
    &= H_1 ( FE + [E,F])
    = H_1 FE + H_1 (H_1 - H_2)
    \\&= \cdots
    = FEH_1 + 2FE + H_1^2 - H_1H_2.
    \qedhere
\end{align*}
\end{example}

\subsection{Central characters and Casimir elements}\label{sec:U-center}

Let $\frakg$ be a Lie algebra. An element $X \in \calU(\frakg)$ of the universal enveloping algebra is \emph{central} if $XY = YX$ for every $Y \in \calU(\frakg)$ (or equivalently every $Y \in \frakg$). The center $\calZ(\frakg)$ of $\calU(\frakg)$ is the set of all central elements; $\calZ(\frakg)$ is a commutative associative subalgebra of $\calU(\frakg)$. 

A representation $\rho \colon \frakg \to \End(V)$ of $\frakg$ extends to a representation of $\calU(\frakg)$, and if $X \in \calZ(\frakg)$ is a central element, the map $\rho(X) : V \to V$ is $\calU(\frakg)$-equivariant. If $V$ is irreducible, then by Schur's lemma (\Cref{lem:schur}) $\rho(X)$ is a multiplication by some scalar $\chi_V(X)$. More precisely, $\rho$ induces a $1$-dimensional representation $\chi_V : \calZ(\frakg) \to \bbC$. This representation is called the \emph{central character} of $V$. 

It is a crucial fact that for Lie algebras $\frakg$ of reductive groups central characters separate non-isomorphic irreducible representations.
\begin{theorem}[\cite{racah-rendlincei}, see also~{\cite[\S VIII.8.5, Cor.2]{Bourbaki-Lie}}]
    \label{theorem:central-characters}
    Let $\frakg$ be a reductive Lie algebra.
    If $V$ and $W$ are irreducible representations of $\frakg$ and the central characters $\chi_V$ and $\chi_W$ coincide, then $V$ and $W$ are isomorphic.
\end{theorem}

It follows that for reductive $\frakg$ the isotypic components of a $\frakg$-representation are also isotypic components of the induced $\mathcal{Z}(\frakg)$-representation.
More specifically, if a $\frakg$-representation $V$ decomposes as $\bigoplus_{\lambda \in \Lambda} V_\lambda^{\oplus m_\lambda}$ with pairwise non-isomorphic $V_\lambda$, then as a $\mathcal{Z}(\frakg)$-representation we have $V \cong \bigoplus_{\lambda \in \Lambda} \chi_\lambda^{\oplus m_\lambda \cdot \dim V_\lambda}$, where we denote by $\chi_\lambda$ the $1$-dimensional representation of $\mathcal{Z}(\frakg)$ corresponding to the central character of $V_\lambda$.

\begin{example}
    Consider the $\gl_3$ representation $V = \bbC[x_1,x_2,x_3]_d$ of homogeneous degree 3 polynomials in three variables (\Cref{claim:shifted partials}). Consider the element $I \coloneqq - (E_{1,1} + E_{2,2} + E_{3,3})$.
    Clearly, $I \in \mathcal{Z}(\gl_3)$. We determine the value of the central characters. Take $f \in V$.
    We find 
    \begin{align*}
        I.f 
        = x_1 \frac{\partial}{\partial x_1} f + x_2 \frac{\partial}{\partial x_2} f +  x_3 \frac{\partial}{\partial x_3} f
        = d f.
    \end{align*}
    So $\chi_{V}(I) = d$.
    As a simple applications of \Cref{theorem:central-characters}, this allows us to conclude $\bbC[x_1,x_2,x_3]_d$ and $\bbC[x_1,x_2,x_3]_{d'}$ contain no isomorphic irreducible subrepresentations unless $d = d'$.
\end{example}

When $\frakg$ is the Lie algebra of a reductive group, then $\mathcal{Z}(\frakg)$ can be described explicitly as the ring of invariants for the the action of the Weyl group on the Cartan subalgebra.
Note that since the Cartan subalgebra $\frakh$ is abelian, the universal enveloping algebra $\mathcal{U}(\frakh)$ is in fact commutative and can be identified with the polynomial algebra $\bbC[\frakh^*]$.

\begin{theorem}[Harish--Chandra isomorphism, {\cite[Ch.~23]{Hum:LieAlg};\cite{Berdjis-CasimirOperators}}]\label{thm: HarishChandra}
    Let $\frakg$ be a reductive algebra and let $\frakh$ be a Cartan subalgebra.
    The center $\mathcal{Z}(\frakg)$ of the universal enveloping algebra is isomorphic to the algebra of invariants $\mathcal{U}(\frakh)^{\calW_\frakg}$. Moreover, if $\dim \frakh = k$, there exist algebraically independent $C_1, \dots, C_k \in \mathcal{U}(\frakg)$ such that $\mathcal{Z}(\frakg) = \mathbb{C}[C_1, \dots, C_k]$. The lengths of these elements coincide with the degrees of the fundamental invariants of $\mathcal{U}(\frakh)^{\calW_\frakg}$.
\end{theorem}
The generators $C_1, \dots, C_k$ of the center $\mathcal{Z}(\frakg)$ are known as \emph{Casimir elements}.
These generators are not unique.

Consider the case of $\frakgl_k$, with $\frakh$ the Cartan subalgebra of diagonal matrices. Then $\calW_{\frakgl_k}$ is the symmetric group $\frakS_k$, and it acts on $\frakh$ by permuting the diagonal elements. Then $\calU(\frakh)$ is the polynomial ring in $k$ variables, and $\calU(\frakh)^{\calW_{\frakgl_k}}$ is the subring of polynomials invariant under permutation of variables: this is the classical ring of symmetric polynomials in $k$ variables. It is generated, for instance, by the elementary symmetric polynomials, which have degrees $1 \vvirg k$. These numbers give lengths of Casimir elements generating the center $\mathcal{Z}(\frakg)$.
There are several constructions of Casimir elements, one possible choice being
\begin{align}
\label{eq:casimir}
C_p \coloneqq \sum_{i_1 = 1}^k \dots \sum_{i_p = 1}^k E_{i_1,i_2} E_{i_2, i_3} \cdots E_{i_p,i_1}
\quad\text{with } 1 \leq p \leq k.
\end{align}
\begin{example}
    \label{example:central-character-and-casimir-elements}
    Let $\frakg = \gl_2$. Then the Casimir elements are given by
    \begin{align*}
        C_1 = E_{1,1} + E_{2,2}
        \qquad\text{and}\qquad 
        C_2 = E_{1,1}E_{1,1} + E_{1,2}E_{2,1} + E_{2,1}E_{1,2} +  E_{2,2}E_{2,2}.
    \end{align*}
  We will use \Cref{claim:metapolynomials} to evaluate the application of $C_1$ and $C_2$ to the metapolynomials
    \begin{align*}
         \Delta = c_{11}^{2} + 2 \, c_{02} c_{20}
         \qquad\text{and}\qquad 
         \Gamma = -c_{11}^{2} + 4 \, c_{02} c_{20}.
    \end{align*}
    We find that 
    \begin{align*}
        C_1.\Delta = E_{1,1}.\Delta + E_{2,2}.\Delta 
        = 2 c_{11}^{2} + 4 c_{02} c_{20} + 2 c_{11}^{2} + 4 c_{02} c_{20} = 4\Delta
    \end{align*}
    and similarly $C_1.\Gamma = 4\Gamma$ (in fact, one can show $C_1$ always scales a metapolynomial of format $(\delta,d,k)$ by $d\delta$). 
    We also compute
    \begin{align*}
        E_{2,1}.\Delta 
            &= 
            2\,c_{02}^2 c_{11} + 2\,c_{11}^2 c_{02}
            + 2\cdot 0c_{20} + 2\,c_{02}c_{11}
            = 6\,c_{02}^2 c_{11} 
            \\
        (E_{1,2}E_{2,1}).\Delta
            &= 6\, c_{11}c_{11} + 6 \cdot 2 \,c_{02}c_{20}
            = 6\Delta
    \end{align*}
    By symmetry of multiindices in the metapolynomial we know $(E_{2,1}E_{1,2}).\Delta = 6\Delta$ too. 
    An analogous computation shows $E_{2,1}.\Gamma = 0$ and $E_{1,2}.\Gamma = 0$.
    We find
    \begin{align*}
        C_2.\Delta 
        = 
        \big(2^2 + 6 + 6 + 2^2\big) \Delta = 20\Delta 
        \qquad\text{and}\qquad
        C_2.\Gamma = 
        \big(2^2 + 0 + 0 + 2^2\big) \Delta = 8\Delta. 
    \end{align*}
    We obtain that $\Delta$ and $\Gamma$ live in different isotypic components with respect to 
    $ \mathcal Z(\frakgl_2)$, and hence with respect to $\gl_2$. 
\end{example}

The eigenvalues of these Casimir elements on irreducible representations of $\frakgl_k$ are computed by Perelomov and Popov~\cite{perelomov-popov}: 
\begin{align}
    \label{eq:casimir formula}
    \chi_\lambda(C_p) = \Tr(A_\lambda^p E)
\end{align}
where $E$ is the $k \times k$ matrix consisting of all ones, and the entries of the $k\times k  $ matrix $A_\lambda = (a^{(\lambda)}_{ij})$ are given by
\begin{align}
a^{(\lambda)}_{ij} = \begin{cases}
    \lambda_i + k - i, & \text{if $i = j$,}\\
    0, & \text{if $i > j$,}\\
    -1, & \text{if $i < j$.}
\end{cases}
\end{align}

\begin{example}
    Observe first that $\Tr(A_\lambda^p E)$ is simply the sum of the entries of $A_\lambda^p$.
    Consider again \Cref{example:central-character-and-casimir-elements}. We compute the formula in \cref{eq:casimir formula}. 
    \begin{align*}
        A_{(4,0)} &= \left[\begin{smallmatrix}
            5 & -1  \\
            0 & 2  \\
        \end{smallmatrix}\right]
        \ \text{ and thus }\ 
        \chi_{(4,0)}(C_2) = \Tr \left(
        \left[\begin{smallmatrix}
            25 & -7  \\
            0 & 4  \\
        \end{smallmatrix}\right]
        \left[\begin{smallmatrix}
            1 & 1  \\
            1 & 1  \\
        \end{smallmatrix}\right]\right)
        = 20
        \\
        A_{(2,2)} &= \left[\begin{smallmatrix}
            3 & -1  \\
            0 & 2  \\
        \end{smallmatrix}\right]
        \ \text{ and thus }\ 
        \chi_{(2,2)}(C_2) = \Tr \left(
        \left[\begin{smallmatrix}
            9 & -5  \\
            0 & 4  \\
        \end{smallmatrix}\right]
        \left[\begin{smallmatrix}
            1 & 1  \\
            1 & 1  \\
        \end{smallmatrix}\right]\right)
        = 8.
    \intertext{
    Observe that this matches our results in \Cref{example:central-character-and-casimir-elements}.
    We give one more example for $\gl_3$.
    }
        A_{(6,0,0)} &= \left[\begin{smallmatrix}
            8 & -1 & -1 \\
            0 & 1 & -1 \\
            0 & 0 & 0
        \end{smallmatrix}\right]
        \ \text{ and thus }\ 
        \chi_{(6,0,0)}(C_2) = \Tr \left(
        \left[\begin{smallmatrix}
            64 & -9 & -7 \\
            0 & 1 & -1 \\
            0 & 0 & 0
        \end{smallmatrix}\right]
        \left[\begin{smallmatrix}
            1 & 1 & 1 \\
            1 & 1 & 1 \\
            1 & 1 & 1
        \end{smallmatrix}\right]\right)
        = 48. \qedhere
    \end{align*}
\end{example}

\subsection{Gelfand--Tsetlin theory}\label{sec:gelfand-tsetlin}

Gelfand--Tsetlin bases are used to explicitly construct irreducible representations of some classical Lie algebras.
The explicit bases first appeared in~\cite{gelfand-tsetlin} without much explanation.
The construction was expanded upon by \v{Z}elobenko~\cite{zhelobenko-spectral} and independently by Baird and Biedenharn~\cite{baird-biedenharn-lie-2}.
It can be explained in terms of restrictions of irreducible representations along a chain of Lie algebra inclusions.

Let $\mathfrak{a}$ be a Lie subalgebra of $\mathfrak{b}$.
Every representation $V$ of $\mathfrak{b}$ can be regarded as a representation of $\mathfrak{a}$ using this inclusion.
The resulting representation is called the \emph{restriction of $V$ to $\mathfrak{a}$} and is denoted as $V{\downarrow}^{\mathfrak{b}}_{\mathfrak{a}}$.

Let $\frakg_1 \subset \frakg_2 \subset \dots \subset \frakg_k$ be an inclusion chain of reductive Lie algebras.
We call this chain a \emph{Gelfand--Tsetlin chain} if $\frakg_1$ is abelian and every irreducible representation $V$ of $\frakg_\ell$ ($1 < \ell \leq k$), when restricted to $\frakg_{\ell - 1}$, is a multiplicity-free completely reducible representation, that is, 
\begin{align}
    V_\lambda{\downarrow}^{\frakg_\ell}_{\frakg_{\ell - 1}} \cong \bigoplus_{\mu \in R_\ell(\lambda)} V_\mu
\end{align}
for some set $R_\ell(\lambda)$ of isomorphism types of irreducible representations of $\frakg_{i - 1}$.
We write $\mu \to \lambda$ if $V_\mu$ appears in the decomposition for $V_\lambda$. 

Let $V_\lambda$ be a irreducible representation of $\frakg_k$.
By repeatedly restricting irreducible representations of $\frakg_i$ we obtain a decomposition of $V_\lambda$ as a representation of $\frakg_1$.
\begin{align}
    V_\lambda{\downarrow}^{\frakg_k}_{\frakg_1} = \bigoplus_{\lambda^{(1)} \to \lambda^{(2)} \to \dots \to \lambda^{(k)} = \lambda} V_{\lambda^{(1)}}
\end{align}
Since $\frakg_1$ is abelian, its irreducible representations are $1$-dimensional, so the decomposition above distinguishes a set of $1$-dimensional subspaces in $V_\lambda$ indexed by chains $\lambda^{(1)} \to \lambda^{(2)} \to \dots \to \lambda^{(k)}$ called \emph{Gelfand--Tsetlin patterns}.
We denote the subspace of $V_{\lambda^{(k)}}$ corresponding to a Gelfand--Tsetlin pattern $T = \lambda^{(1)} \to \lambda^{(2)} \to \dots \to \lambda^{(k)}$ by $V_T$. 

The inclusions $\frakg_\ell \subseteq \frakg_k$ give an inclusion of universal enveloping algebras $\calU(\frakg_\ell) \subseteq \calU(\frakg_k)$, and therefore a chain $\mathcal{U}(\frakg_1) \subset \mathcal{U}(\frakg_2) \subset \dots \subset \mathcal{U}(\frakg_k)$. The  \emph{Gelfand--Tsetlin algebra} $\GZ(\frakg_k)$ of $\frakg_k$ is the subalgebra of $\calU(\frakg)$ generated by all centers $\calZ(\frakg_\ell)$ with $\ell \leq k$. Since the center~$\mathcal{Z}(\frakg_\ell)$ of $\mathcal{U}(\frakg_\ell)$ commutes with all elements of $\mathcal{U}(\frakg_\ell)$, it commutes in particular with the centers~$\mathcal{Z}(\frakg_j)$ for $j < \ell$; as a result, $\GZ(\frakg_k)$ is a commutative subalgebra.

Consider a Gelfand--Tsetlin pattern $T = \lambda^{(1)} \to \lambda^{(2)} \to \dots \to \lambda^{(k)}$ and the corresponding $1$-dimensional subspace $V_T \subset V_{\lambda^{(k)}}$.
For every $\ell \leq k$ the subspace $V_T$ is contained in an irreducible representation of $\frakg_\ell$ isomorphic to $V_{\lambda^{(\ell)}}$, and the center $\mathcal{Z}(\frakg_\ell)$ acts on $V_T$ via the central character $\chi_{\lambda^{(\ell)}}$.
It follows that $V_T$ is a $1$-dimensional representation of $\GZ(\frakg_k)$ given by a \emph{Gelfand--Tsetlin character} $\chi_T \colon \GZ(\frakg_k) \to \mathbb{C}$ such that $X . v = \chi_T(X) v$ for $X \in \GZ(\frakg_k)$, $v \in V_T$.
More specifically, on $\mathcal{Z}(\frakg_\ell)$ we have $\chi_T(X) = \chi_{\lambda^{(\ell)}}(X)$.
\Cref{theorem:central-characters} implies that different Gelfand--Tsetlin patterns correspond to different Gelfand--Tsetlin characters.

By the previous discussion, every completely reducible representation of $\frakg_k$ is also completely reducible as a representation of $\GZ(\frakg_k)$.
That is, if $V$ is a completely reducible representation of $\frakg_k$ with isotypic decomposition $V = \bigoplus_{\lambda \in \Lambda} V_\lambda^{\oplus m_\lambda}$, then the action of $\GZ(\frakg_k)$ on $V$ splits the decomposition further into Gelfand--Tsetlin characters as
\begin{align}
V = \bigoplus\limits_{T \colon \lambda^{(1)} \to \dots \to \lambda^{(k)}, \lambda^{(k)} \in \Lambda} \chi_T^{\oplus m_{\smash{\lambda^{(k)}}}}.
\end{align}
We call the components of this decomposition \emph{GZ-isotypic} spaces of $V$.

The classical case ---used to construct a basis of irreducible representations of $\frakgl_n$--- corresponds to the chain of inclusions $\frakgl_1 \subset \frakgl_2 \subset \dots \subset \frakgl_k$ where $\frakgl_{k - 1}$ is included in $\frakgl_k$ as the subalgebra of $\frakgl_k$ consisting of matrices having the $k$-th row and the $k$-th column identically equal to $0$. The Gelfand--Tsetlin algebra in this case is the polynomial algebra $\GZ(\frakg_k) = \mathbb{C}[C_{\ell,p}]$ generated by the Casimir elements 
\begin{align}
    \label{eq:casimirGZ}
    C_{\ell,p} \coloneqq \sum_{i_1 = 1}^\ell \dots \sum_{i_p = 1}^\ell E_{i_1,i_2} E_{i_2, i_3} \cdots E_{i_p,i_1}
    \quad \text{with } 1 \leq \ell \leq k \text{ and } 1 \leq p \leq \ell.
\end{align}

The irreducible representations of $\frakgl_k$ correspond highest weights $\lambda = (\lambda_1, \dots, \lambda_k)$ with integer $\lambda_1 \geq \dots \geq \lambda_k$.
The restriction $V_\lambda{\downarrow}^{\frakgl_k}_{\frakgl_{k - 1}}$ is given by the classical branching rule (attributed to Schur in~\cite{molev}) which states that $\mu \to \lambda$ if and only if $\mu$ and $\lambda$ \emph{interleave}, that is $\lambda_\ell \geq \mu_\ell \geq \lambda_{\ell + 1}$ for all $\ell$, $1 \leq \ell \leq k - 1$.
Moreover, $E_{k,k}$ acts on the copy of $V_\mu$ in $V_\lambda$ as multiplication by $|\lambda| - |\mu| = \sum_{\ell = 1}^{k - 1} (\lambda_k - \mu_k)$.

A \emph{Gelfand-Tsetlin pattern} is given by a sequence $\lambda^{(1)} \to \lambda^{(2)} \to \dots \to \lambda^{(k)}$, that is, by integers~$\lambda^{(\ell
)}_i$ with $1 \leq \ell \leq k$, $1 \leq i \leq \ell$ such that $\lambda^{(\ell)}_i \geq \lambda^{(\ell - 1)}_i \geq \lambda^{(\ell)}_{i + 1}$, often arranged into a triangular shape, see \Cref{fig:chainExample1}.

\setlength{\unitlength}{30pt}
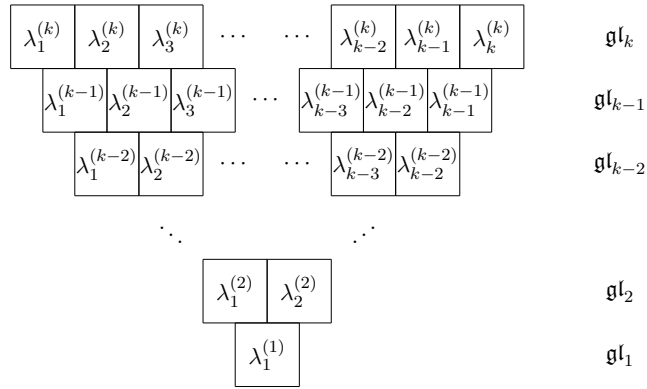
\begin{figure}[!htbp]
\begin{center}
\begin{displaymath}
\scalebox{.8}{
\begin{picture}(10,6)

\put(9,5){\ylabel{$\frakgl_{k}$}}

\put(0,5){\ylabel{$\lambda_1^{(k)}$}}
\put(0,5){\ybox}
\put(1,5){\ylabel{$\lambda_2^{(k)}$}}
\put(1,5){\ybox}
\put(2,5){\ylabel{$\lambda_3^{(k)}$}}
\put(2,5){\ybox}

\put(3,5){\ylabel{$\cdots$}}
\put(4,5){\ylabel{$\cdots$}}

\put(5,5){\ylabel{$\lambda_{k-2}^{(k)}$}}
\put(5,5){\ybox}
\put(6,5){\ylabel{$\lambda_{k-1}^{(k)}$}}
\put(6,5){\ybox}
\put(7,5){\ylabel{$\lambda_{k}^{(k)}$}}
\put(7,5){\ybox}

\put(9,4){\ylabel{$\frakgl_{k-1}$}}

\put(0.5,4){\ylabel{$\lambda^{(k-1)}_{1}$}}
\put(0.5,4){\ybox}
\put(1.5,4){\ylabel{$\lambda^{(k-1)}_{2}$}}
\put(1.5,4){\ybox}
\put(2.5,4){\ylabel{$\lambda^{(k-1)}_{3}$}}
\put(2.5,4){\ybox}

\put(3.5,4){\ylabel{$\cdots$}}

\put(4.5,4){\ylabel{$\lambda^{(k-1)}_{k-3}$}}
\put(4.5,4){\ybox}
\put(5.5,4){\ylabel{$\lambda^{(k-1)}_{k-2}$}}
\put(5.5,4){\ybox}
\put(6.5,4){\ylabel{$\lambda^{(k-1)}_{k-1}$}}
\put(6.5,4){\ybox}

\put(9,3){\ylabel{$\frakgl_{k-2}$}}

\put(1,3){\ylabel{$\lambda^{(k-2)}_{1}$}}
\put(1,3){\ybox}
\put(2,3){\ylabel{$\lambda^{(k-2)}_{2}$}}
\put(2,3){\ybox}

\put(3,3){\ylabel{$\cdots$}}
\put(4,3){\ylabel{$\cdots$}}

\put(5,3){\ylabel{$\lambda^{(k-2)}_{k-3}$}}
\put(5,3){\ybox}
\put(6,3){\ylabel{$\lambda^{(k-2)}_{k-2}$}}
\put(6,3){\ybox}

\put(2,2){\ylabel{$\ddots$}}
\put(5,2){\ylabel{$\iddots$}}

\put(9,1){\ylabel{$\frakgl_{2}$}}

\put(3,1){\ylabel{$\lambda^{(2)}_{1}$}}
\put(3,1){\ybox}
\put(4,1){\ylabel{$\lambda^{(2)}_{2}$}}
\put(4,1){\ybox}

\put(9,0){\ylabel{$\frakgl_{1}$}}
\put(3.5,0){\ylabel{$\lambda^{(1)}_{1}$}}
\put(3.5,0){\ybox}
\end{picture}
}
\end{displaymath}
\caption{Gelfand--Tsetlin pattern of $\frakgl_k$}
\label{fig:chainExample1}
\end{center}
\end{figure}

The $1$-dimensional subspace $V_T$ corresponding to the pattern $T = \lambda^{(1)} \to \lambda^{(2)} \to \dots \to \lambda^{(k)}$ is spanned by a weight vector of weight $(|\lambda^{(1)}|, |\lambda^{(2)}|-|\lambda^{(1)}|, \dots, |\lambda^{(k)}| - |\lambda^{(k - 1)}|)$.
The generators $C_{\ell, p}$ of $\GZ(\frakgl_k)$ act on this space as multiplications by the corresponding eigenvalues $\chi_T(C_{\ell, p})$.
Since $C_{\ell, p}$ is a Casimir operator for~$\frakgl_\ell$, we have $\chi_T(C_{\ell, p}) = \chi_{\lambda^{(\ell)}}(C_{\ell, p}) = \tr(A_{\lambda^{(\ell)}}^p E)$ from Perelomov--Popov formula.

If $\lambda^{(k)}$ is a partition, that is, all $\lambda^{(k)}_i$ are nonnegative, then the Gelfand--Tsetlin patterns for $\frakgl_k$ are in bijective correspondence with semistandard Young tableaux filled with numbers from $\{1, 2, \dots, k\}$. Recall that a Young tableau is semistandard whenever the numbers are strictly increasing along the columns, and weakly increasing along the rows.

A semistandard tableau $T$ gives rise to a sequence of partitions $\lambda^{(1)}, \dots, \lambda^{(k)}$ where $\lambda^{(\ell)}$ is the shape of the tableau $T^{(\ell)}$ obtained from $T$ be removing all boxes labeled with numbers greater than $\ell$.
Semistandardness of $T$ guarantees that all $T^{(\ell)}$ are also semistandard tableaux and the interleaving property $\lambda^{(\ell)}_i \geq \lambda^{(\ell - 1)}_i \geq \lambda^{(\ell)}_{i + 1}$ holds (see e.\,g.~\cite[p.5]{Macdonald-book}).
The weight is given by the content of the tableau $T$, that is, the vector $\mu = (\mu_1, \dots, \mu_k)$ where $\mu_i$ records the number of occurrences of the label $i$ in the tableau.

For example, the tableau $T = {\young(1122,23,44)}$ gives
$$T^{(1)} = \young(11), \quad T^{(2)} = \young(1122,2), \quad T^{(3)} = \young(1122,23), \quad T^{(4)} = T = \young(1122,23,44),$$
and the sequence of shapes is $\lambda^{(1)} = (2)$, $\lambda^{(2)} = (4,1)$, $\lambda^{(3)} = (4,2,0)$, $\lambda^{(4)} = (4,2,2,0)$.
The corresponding weight is $(2, 3, 1, 2)$, the content of the tableau.

The superstandard tableau of shape $\lambda = (\lambda_1, \dots, \lambda_n)$ corresponds to the Gelfand-Tsetlin pattern with $\lambda^{(\ell)} = (\lambda_1, \dots, \lambda_\ell)$.
The corresponding $1$-dimensional subspace of $V_\lambda$ contains the highest weight vectors.
To prove this, recall that since $E_{i, j}$ with $i \neq j$ span root spaces of $\frakgl_k$, they shift the weights of weight vectors, that is, $E_{i,j} . v$ is a weight vector of weight $\mu + e_i - e_j$.
The content of the superstandard tableau of shape $\lambda$ is $\lambda$, so every nonzero vector $v$ in the corresponding $1$-dimensional subspace has weight $\lambda$, and for every $E_{i,j}$  the vector $E_{i,j} . v$ has weight $\lambda + e_i - e_j$. But there are no semistandard tableaux with shape $\lambda$ and content $\lambda + e_i - e_j$ when $i < j$.
Therefore, $E_{i,j} . v = 0$ if $i < j$, and it follows that $v$ is annihilated by every strictly upper triangular matrix.

\section{Efficient construction of projectors}
\label{sec:effconproj}
In this section we prove \Cref{thm:main}.
We explain the construction of the projectors in \Cref{sec:projector construction}, and describe the efficient circuit implementation in \Cref{subsec:ckttransform}.

\subsection{Construction of projectors}
\label{sec:projector construction}

We construct the required projections as projections onto a common eigenspace of several commuting matrices. The construction is elementary and is based on the classical construction of eigenprojectors from basic linear algebra.

We present the construction in full generality in \Cref{thm:projector}. The projectors in the cases relevant for \Cref{thm:main} are obtained by applying the general construction to suitable commutative subalgebras of the universal enveloping algebra $\calU(\frakgl_k)$, see \Cref{table:projection-algebras-generators-and-diversity}. If $\mathcal{A}$ is a commutative algebra generated by $U_1, \dots, U_q$, then a representation of $\rho : \calA \to \End(V)$ is uniquely determined by $\rho(U_1) \vvirg \rho(U_q)$, which are $q$ commuting endomorphisms of $V \to V$.
Every irreducible representation of $\mathcal{A}$ is $1$-dimensional, so it is given by an algebra homomorphism $\chi \colon \mathcal{A} \to \mathbb{C}$.
In other words, an irreducible representation is spanned by a common eigenvector of the generators $U_1, \dots, U_q$ of $\mathcal{A}$ with the corresponding eigenvalues $\chi(U_1), \dots, \chi(U_q)$.

Recall that a representation is called \emph{completely reducible} if it can be decomposed into a direct sum of irreducible representations.
All representations in this paper have this property.
For a commutative algebra this means that the generators of the algebra act not just by commuting but by simultaneously diagonalizable linear maps.

We slightly abuse notation and denote the character $\chi \colon \mathcal{A} \to \bbC$ and the corresponding irreducible representation by the same letter.
Every completely reducible representation $V$ of $\mathcal{A}$ has a unique \emph{isotypic decomposition}
\begin{align}
V \cong \bigoplus_{\lambda \in \Lambda} \chi_\lambda^{\oplus m_\lambda},
\end{align}
where $\Lambda$ is a set indexing irreducible representations of $\mathcal{A}$.
We denote the isotypic component $\chi_\lambda^{\oplus \lambda}$ in $V$ by $V_{\mathcal{A},\lambda}$.

The length of the projectors as elements of $\mathcal{U}(\frakg)$ depends mainly on the number of different isotypic components of the representation.
\begin{definition}
    Let $V$ be a completely reducible representation of an algebra $\mathcal{A}$.
    We call the number of isomorphism types of irreducible subrepresentations of $V$ the \emph{diversity} of $V$.
\end{definition}

\begin{theorem}\label{thm:projector}
    Let $\mathcal{A}$ be a commutative algebra generated by $U_1, \dots, U_q$ and $V$ be a completely reducible representation of $\mathcal{A}$ with diversity at most $p$.
    For every irreducible representation $\chi_\lambda$ of $\mathcal{A}$ there exists a polynomial~$f_\lambda$ of degree at most $p - 1$ such that 
    \begin{align}
    f_\lambda(U_1, \dots, U_q) . x = \Pi_\lambda x
    \end{align}
    where $\Pi_\lambda \colon V \to V$ is the unique $\calA$-equivariant projection onto the isotypic component $V_{\mathcal{A},\lambda}$.
\end{theorem}
\begin{proof}
    Let $\Lambda$ be the set of isomorphism types appearing in the isotypic decomposition of $V$, so that $V \cong \bigoplus_{\mu \in \Lambda} V_{\mathcal A,\mu}$ with $V_{\mathcal A,\mu} \neq 0$. An element $U \in \mathcal{A}$ acts as $U. x = \chi_\mu(U) x$ on $V_{\mathcal A,\mu}$. If $\lambda \notin \Lambda$, set $f_\lambda = 0$.

    Fix $\lambda \in \Lambda$. For every $\mu \in \Lambda$ with $\mu \neq \lambda$, let $U_{i_\mu}$ be a generator such that $\chi_\mu ( U_{i_\mu} ) \neq \chi_\lambda(U_{i_\mu})$. Such generator exist because $\chi_\lambda$ and $\chi_\mu$ are different characters of $\mathcal{A}$. Define 
    \begin{align}
        P_\lambda \coloneqq \prod_{\mu \in \Lambda, \mu \neq \lambda} \frac{U_{i_\mu} - \chi_\mu(U_{i_\mu})}{\chi_\lambda(U_{i_\mu}) - \chi_\mu(U_{i_\mu})} 
        \quad\in \mathcal{A}.
    \end{align}
 The factor $\dfrac{U_{i_\mu} - \chi_\mu(U_{i_\mu})}{\chi_\lambda(U_{i_\mu}) - \chi_\mu(U_{i_\mu})}$ acts on $V$ as the zero map on $V_{\mathcal A,\mu}$ and as the identity map on~$V_{\mathcal A,\lambda}$; moreover, its action maps isotypic components to themselves. The element $P_\lambda$ acts as the composition of these factors; therefore its image in $\End(V)$ is the projection $\Pi_\lambda$ onto $V_{\mathcal A, \lambda}$. Finally, since every factor is linear in $U_1 \vvirg U_q$, $P_\lambda$ is a polynomial in $U_1 \vvirg U_q$. Its degree is (at most) $p-1$.
\end{proof}

The projections required in~\Cref{thm:main} arise as special cases of this construction.
In each case, the image of the projection is an isotypic component of some commutative subalgebra $\mathcal{A} \subset \mathcal{U}(\frakgl_k)$.
\Cref{thm:projector} then gives a construction of the projectors as elements of $\mathcal{U}(\frakgl_k)$, expressed as polynomials in the generators of $\mathcal{A}$.
The length of the projectors is easily bounded in terms of lengths of generators of $\mathcal{A}$ in $\mathcal{U}(\gl_K)$ and the diversity of $\bbC[\bbC[x_1,\ldots,x_k]_d]_\delta$ as a representation of $\mathcal{A}$.
Indeed, if the generators $U_1, \dots, U_q$ of $\mathcal{A}$ have length at most $\ell$, and the projectors are expressed as polynomials in $U_1, \dots, U_q$ of degree at most $p - 1$, then the lengths of the projectors are bounded by $\ell(p - 1)$. The following table summarizes these parameters for the three cases considered.

\begin{table}[H]
    \centering
    {\renewcommand{\arraystretch}{1.2}
    \begin{tabular}{c||c|c|c}
        Isotypic components   & Weight spaces & $\lambda$-isotypic spaces & $T$-isotypic spaces \\\hline\rule{0pt}{1.2em}%
        Algebra $\mathcal{A}$ & $\mathcal{U}(\frakh) \cong \mathbb{C}[\frakh^*]$ & $\mathcal{Z}(\frakgl_k)$ & $\GZ(\frakgl_k)$ \\
        Generators $U_i$ &  $\big\{E_{i,i} \mid i \in [k]\big\}$  &  $\big\{C_{p} \mid p \in [k]\big\}$ &  $\big\{C_{\ell,p} \mid \ell \in [k], p \in [\ell]\big\}$  \\
        Generator lengths  & $1$ & $\leq k$ & $\leq k$ \\
        Diversity &   $\leq (\delta d)^k$  &  $\leq (\delta d)^k$  &  $\leq (\delta d)^{k^2}$
    \end{tabular}}
       \caption{
    Isotypic components, the generators of their corresponding subalgebras, and their diversity in $\bbC[\bbC[x_1,\ldots,x_k]_d]_\delta$.
    See \cref{eq:casimir} and \cref{eq:casimirGZ} for the definitions of $C_p$ and $C_{\ell,p}$.
    }    \label{table:projection-algebras-generators-and-diversity}
\end{table}

\begin{remark}
In the construction of \Cref{thm:projector} we require knowledge of the value of $\chi_\mu$ evaluated on the generators $U_i$, for all irreducible representation types $\mu$ occurring in the representation. 
For the weight spaces these values are given by the weights themselves, and for the other two algebra they can be computed using \cref{eq:casimir formula}. 
It is then also possible to remove any choices in the construction by taking for each term instead a linear combination $U$ of all $U_i$ with algebraically independent irrational weights. 
Since the values $\chi_\mu$ are always integer, this ensures $\chi_\mu(U) \neq \chi_\lambda(U)$ when $\mu \neq \lambda$.
\end{remark}

\begin{example}
\label{example:projector construction}
We construct and evaluate a $\lambda$-isotypic projector for metapolynomials of format $(3,2,3)$. 
From \Cref{example:irreps-of-polynomials-2-3-and-metapolynomials-3-2-3} we know $\Lambda = \{(6,0,0),(4,2,0),(2,2,2)\}$.
By \cref{eq:casimir formula} we find that the Casimir element $C_2$ scales these three irreducible representations with different scalars, namely $48$, $28$ and $12$ respectively.
Hence we can use $C_2$ to construct the projectors from \Cref{thm:projector}. 
Consider first the projector $P_{(6,0,0)}$ to the $(6,0,0)$-isotypic component. 
We have that
\begin{align*}
    P_{(6,0,0)} = 
    \bigg(\frac{C_2 - \chi_{(2,2,2)}(C_2)}{\chi_{(6,0,0)}(C_2) - \chi_{(2,2,2)}(C_2)}\bigg)
    \bigg(\frac{C_2 - \chi_{(4,2,0)}(C_2)}{\chi_{(6,0,0)}(C_2) - \chi_{(4,2,0)}(C_2)}\bigg)
    =
    \bigg(\frac{C_2 - 12}{36}\bigg)
    \bigg(\frac{C_2 - 28}{20}\bigg)
    .
\end{align*}
Consider again the example $\Delta = c_{002}c_{020}c_{200}$ from the introduction. 
We will compute $P_{(6,0,0)}.\Delta$. Using the equations in \Cref{claim:metapolynomials} we find
\begin{align*}
    C_2.\Delta &
    = \sum_{i=1}^3 \sum_{j=1}^3 \big( E_{i,j}E_{j,i} \big).\Delta
    = 2 \, c_{020} c_{101}^{2} + 2 \, c_{002} c_{110}^{2} + 2 \, c_{011}^{2} c_{200} + 24 \, c_{002} c_{020} c_{200}.
\end{align*}
It directly follows that
\begin{align*}
    C_2.\Delta - 28\Delta
    =
    2 \, c_{020} c_{101}^{2} + 2 \, c_{002} c_{110}^{2} + 2 \, c_{011}^{2} c_{200} - 4 \, c_{002} c_{020} c_{200}.
\end{align*}
We apply $C_2$ again (at this point a computer algebra program is recommended) and find
\begin{align*}
    C_2.(C_2.\Delta - 28\Delta) =    
    48 \big(c_{020} c_{101}^{2} + c_{011} c_{101} c_{110} + c_{002} c_{110}^{2} + c_{011}^{2} c_{200} \big).
\end{align*}
Hence, 
\begin{align*}
    60\,P_{(6,0,0)}.\Delta 
    &=
    \frac{C_2.(C_2.\Delta - 28\Delta) - 12(C_2.\Delta - 28\Delta)}{12} 
    \\&=
    2 \, c_{020} c_{101}^{2} + 4 \, c_{011} c_{101} c_{110} + 2 \, c_{002} c_{110}^{2} + 2 \, c_{011}^{2} c_{200} + 4 \, c_{002} c_{020} c_{200}.
\end{align*}
This is indeed the isotypic component as given in the introduction.
The other two can be determined by an analogous computation.
\end{example}

\subsubsection{Weight projectors}

Let $V$ be a representation of a complex reductive Lie group $G$.
The action of $G$ on $V$ induces the action of the corresponding Lie algebra $\frakg$.
Fix a maximal torus $T \subset G$ and let $\frakh \subset \frakg$ be the corresponding Cartan subalgebra.
The torus $T$ is reductive, and the isotypic components of $V$ with respect to the action of $T$ are exactly the weight spaces.
The action of $T$ induces the action of $\frakh$ with the same isotypic components, so the action of $\frakh$ is completely reducible.

Let $H_1, \dots, H_r$ be a basis of the Cartan subalgebra $\frakh$.
Since $\frakh$ is abelian, the universal enveloping algebra $\mathcal{U}(\frakh) \subset \mathcal{U}(\frakg)$ is commutative.
It can be identified with the polynomial algebra $\mathbb{C}[H_1, \dots, H_r]$.
We can now apply~\Cref{thm:projector} to obtain for every weight $\mu \in \frakh^*$ an element $H_\mu \in \mathcal{U}(\frakg)$ such that $H_\mu . v$ is the projection of $v \in V$ onto the weight space of weight $\mu$.
The projectors $H_\mu$ are obtained as polynomial expressions $H_\mu = f_\mu(H_1, \dots, H_r)$ with $\deg f_\mu \leq p - 1$  where $p$ is the diversity of $V$ as a representation of $\frakh$ (or, equivalently, of $\mathcal{U}(\frakh)$), that is, the number of weights of $V$.
Since $H_i$ as elements of $\mathcal{U}(\frakg)$ have length $1$, the length of $H_\mu$ is also bounded by $p - 1$.

We can specialize this construction to the action of $\gl_k$ on the space of metapolynomials $W \coloneqq \bbC[\bbC[x_1,\ldots,x_k]_d]_\delta$. 
The standard Cartan subalgebra consists of all diagonal matrices, and we can take $H_i \coloneqq E_{i,i}$.
The weights appearing in $W$ correspond to $k$-tuples of nonnegative integers summing to $\delta d$ (see \Cref{example:metapolynomials-3-2-3-weights}). The number of such tuples is $\binom{\delta d + k - 1}{k-1} \leq (\delta d)^k$.

\begin{corollary}
\label{corollary:weight-projector}
    For each weight $\mu = (\mu_1, \dots, \mu_k)$ with $\sum_{i = 1}^k \mu_i = \delta d$ there is an element $H_\mu \in \mathcal{U}(\frakgl_k)$ of length at most $(\delta d)^k$ which acts on $\bbC[\bbC[x_1,\ldots,x_k]_d]_\delta$ as the projector onto the weight space corresponding to $\mu$.
\end{corollary}

\subsubsection{Isotypic projectors}

Let $V$ be a representation of a complex reductive Lie group $G$ with the corresponding Lie algebra $\frakg$.
Then $V$ is completely reducible as a representation of $\frakg$, with the same isotypic components as for $G$.
As discussed in~\Cref{sec:U-center}, they are also isotypic components of $V$ as a representation of $\mathcal{Z}(\frakg)$, which is a commutative algebra generated by Casimir elements $C_1, \dots, C_r$.

Applying~\Cref{thm:projector}, we obtain for every highest weight $\lambda$ a corresponding element $Z_\lambda \in \mathcal{U}(\frakg)$ such that $Z_\lambda . v$ is the projection of $v$ onto the isotypic component of corresponding to the irreducible representation with highest weight $\lambda$.
The projectors $Z_\lambda$ are obtained as polynomial expressions $Z_\lambda = f_\lambda(C_1, \dots, C_r)$ with $\deg f_\lambda$ is at most $p - 1$, where $p$ is the diversity of $V$ as a representation of $\frakg$.
It follows that if the Casimir elements $C_1, \dots, C_r$ have length at most $\ell$, then the length of $Z_\lambda$ is at most $\ell(p - 1)$.

For the Lie algebra $\frakgl_k$ Casimir elements have length at most $k$.
Since every highest weight is a weight, the diversity of $\bbC[\bbC[x_1, \dots, x_k]_d]_\delta$ is bounded by $(d\delta)^k$ as in the previous section.

\begin{corollary}
\label{corollary:isotypic-component-projector}
    For each partition $\lambda \partinto[k] {\delta d}$ there is an element $Z_\lambda \in \mathcal{U}(\frakgl_k)$ of length at most $k (\delta d)^k$ which acts on $\bbC[\bbC[x_1,\ldots,x_k]_d]_\delta$ as the projector onto the isotypic component corresponding to $\lambda$.
\end{corollary}

\subsubsection{Highest weight projectors}

Since the highest weight subspace of weight $\lambda$ is exactly the weight $\lambda$ subspace of the isotypic component corresponding to the irreducible representation $V_\lambda$, the projector onto the highest weight subspace is the composition $H_\lambda Z_\lambda$ of the projectors constructed in the previous sections.

\begin{corollary}
\label{corollary:highest-weight-projector}
    For each partition $\lambda \partinto[k]{\delta d}$ there is an element $X_\lambda = H_\lambda Z_\lambda \in \mathcal{U}(\frakgl_k)$ of length at most $(k + 1) (\delta d)^k$ which acts on $\bbC[\bbC[x_1,\ldots,x_k]_d]_\delta$ as the projector onto the highest weight space of weight $\lambda$.
\end{corollary}

\subsubsection{GZ-isotypic projectors}

Let $V$ be a representation of a complex reductive Lie group $G$ with the corresponding Lie algebra $\frakg$ admitting a Gelfand--Tsetlin chain $\frakg_1 \subset \frakg_2 \subset \dots \subset \frakg_k = \frakg$.
As discussed in \Cref{sec:gelfand-tsetlin}, under the action of the Gelfand--Tsetlin algebra $\GZ(\frakg)$ the representation $V$ decomposes into GZ-isotypic components indexed by Gelfand--Tsetlin patterns, and~\Cref{thm:projector} can be used to construct for every Gelfand--Tsetlin pattern $T$ an element $Y_T \in \mathcal{U}(\frakg)$ such that $Y_T.x$ is the projection of $x$ onto the $T$-isotypic subspace of $V$.
The length of $Y_T$ is bounded in terms of the lengths of Casimir elements of $\frakg_\ell$ and the number of Gelfand--Tsetlin patterns appearing in the decomposition of $V$.

Consider the case of $\frakgl_k$ acting on $\bbC[\bbC[x_1, \dots, x_k]_d]_\delta$.
For the chain $\frakgl_1 \subset \frakgl_2 \subset \dots \subset \frakgl_k$, Casimir elements have length at most $k$.
Every irreducible subrepresentation of $\bbC[\bbC[x_1, \dots, x_k]_d]_\delta$ corresponds to a partition $\lambda \partinto[k]{d\delta}$.
Gelfand--Tsetlin patterns are in bijective correspondence with semistandard Young tableaux with $\delta d$ cells and $k$ rows.
Each row of a tableau is an ordered tuple of integers from $\{1, \dots, k\}$, so there are at most $\binom{d\delta + k - 1}{k - 1} \leq (d\delta)^k$ tuples that can serve as rows of a tableau, and therefore at most $(d\delta)^{k^2}$ possible semistandard tableaux.

\begin{corollary}
\label{corollary:semistandard-tableau-projector}
    For each semistandard tableau $T$ of shape $\lambda \partinto[k]{\delta d}$ there is an element $Y_T \in \mathcal{U}(\frakgl_k)$ of length at most $k (\delta d)^{k^2}$ which acts on $\bbC[\bbC[x_1,\ldots,x_k]_d]_\delta$ as the projector onto the $T$-isotypic space corresponding to $T$.
\end{corollary}

\subsection{Circuit transformations}
\label{subsec:ckttransform}

We now combine the bounds on the lengths of the projectors with the PBW theorem (\Cref{thm:pbw}) to prove \Cref{thm:main}.
Pick any ordered basis $(X_1,\ldots,X_{\varcount^2})$ for $\gl_k$, such as $\{E_{i,j}\}_{i,j}$ with $(i,j)$ ordered lexicographically. 

We first introduce some notation.
Set $\basissize \coloneqq k^2$.
For $\bm{i} = (i_1,\ldots,i_{\basissize}) \in \IN^{\basissize}$ we will use the shorthand notation $X^{\bm{i}} \coloneqq X_1^{i_1} \cdots X_\basissize^{i_\basissize} \in \mathcal{U}(\gl_k)$.
Let $\bm{i}' \leq \bm{i}$ denote that $\bm{i}' = (i'_1,\ldots,i'_\basissize) \in \IN^{\basissize}$ with $i'_r \leq i_r$ for all $r \in [\basissize]$. 
Write $\bm{i} - \bm{i}' \coloneqq (i_1-i'_1, \ldots, i_\basissize- i'_\basissize)$ and $\binom{\bm{i}}{\bm{i}'} \coloneqq \binom{i_1}{i'_1} \cdots \binom{i_\basissize}{i'_\basissize}$. Lastly, define $|\bm{i}| \coloneqq \sum_r i_r$ 
and $\boundedtuples \coloneqq \{\bm{i} \in \IN^{\basissize} \mid |\bm{i}| \leq \projwordsize\}$.

By \Cref{claim:metapolynomials}, any element $X \in \mathcal{U}(\gl_k)$ acts as a derivation on the space of metapolynomials $\bbC[\bbC[x_1,\ldots,x_k]_d] = \bigoplus_{\delta \geq 0} \bbC[\bbC[x_1,\ldots,x_k]_d]_\delta$. 

\begin{lemma}
    \label{lemma:circuit-projection}
    Suppose $\Phi,\Gamma \in \bbC[\bbC[x_1,\ldots,x_k]_d]$ are metapolynomials. Then
    \begin{enumerate}
        \item
        \label{item:universal-enveloping-algebra-act-on-sum}
            $X^{\bm{i}}.(\Phi+\Gamma)
            = X^{\bm{i}}.\Phi + X^{\bm{i}}.\Gamma$
        \item
        \label{item:universal-enveloping-algebra-act-on-product}
            $\displaystyle
            X^{\bm{i}}. (\Phi \Gamma)
            =
            \sum_{{\bm{i}}'\colon {\bm{i}}' \leq {\bm{i}}}
            \mbinom{{\bm{i}}}{{\bm{i}}'} 
            \big( X^{{\bm{i}}'}.\Phi \big) \big( X^{{\bm{i}}-{\bm{i}}'}.\Gamma \big) 
            $
    \end{enumerate}
\end{lemma}
\begin{proof}
    \Cref{item:universal-enveloping-algebra-act-on-sum} follows directly from linearity. 
    Since $\gl_k$ acts by derivations (see \Cref{claim:metapolynomials}) we have $X.(fg) = (X.f) g + f (X.g)$.
    For $i \in \IN$ a straightforward induction argument shows $X^i.(fg) = \sum_{i' \colon i' \leq i} \binom{i}{i'} \big(X^{i'}.f\big)\big(X^{i-i'}.g\big)$. Then \Cref{item:universal-enveloping-algebra-act-on-product} follows. 
\end{proof}

\begin{theorem}
\label{proposition:circuit-projection}
    Let $P \in \mathcal{U}(\gl_k)$ be an element of length $L$.
    Suppose there exists an arithmetic circuit $C$ of size $s$ computing a metapolynomial 
    $\Delta \colon \bbC[x_1,\ldots,x_\varcount]_\ideg \to \bbC$. 
    Then there exists an arithmetic circuit $C'$ of size $O\big(sL^{2k^2}\big)$ computing $P.\Delta$.
\end{theorem}
\begin{proof}
    Again set $\basissize \coloneqq k^2$.
    Construct $C'$ from $C$ using the following recursive procedure:
    \begin{enumerate}
        \item Add for every input gate $\Theta$ 
        in $C$ the input gates
            $X^{\bm{i}}.\Theta$ for ${\bm{i}} \in \boundedtuples$.
        \item Add for every addition gate $\Theta = \Phi + \Gamma$ in $C$ the addition gates computing
            $X^{\bm{i}}.\Theta$ for ${\bm{i}} \in \boundedtuples$, using the expression in \Cref{lemma:circuit-projection}.
        \item \label{step:glk-product-gate}
            Add for every product gate $\Theta = \Phi\Gamma$ in $C$:
            \begin{itemize}
                \item product gates computing
                $\big(X^{{\bm{i}}'}.\Phi\big)\big(X^{{\bm{i}}-{\bm{i}}'}.\Gamma\big)$ 
                for ${\bm{i}} \in \boundedtuples$ and ${\bm{i}}' \leq {\bm{i}}$,
                \item addition gates computing 
                $X^{\bm{i}}.\Theta$ for ${\bm{i}} \in \boundedtuples$, using the expression in \Cref{lemma:circuit-projection}. 
            \end{itemize}
        \item 
        The resulting circuit computes $X^{\bm{i}}.\Delta$ for ${\bm{i}} \in \boundedtuples$. 
        By the PBW theorem (\Cref{thm:pbw}), there exists scalars $\beta_{\bm{i}} \in \bbC$ such that 
        \begin{align*}
            P = \sum_{\bm{i} \in \boundedtuples} \beta_{\bm{i}} X^{\bm{i}}.
        \end{align*}
        Add to $C'$ the addition gates computing $P.\Delta$.   
    \end{enumerate}
    The largest blow-up occurs in step \ref{step:glk-product-gate}.
    Let $
        A \coloneqq \{ ({\bm{i}}, {\bm{i}}') \mid {\bm{i}} \in \boundedtuples, {\bm{i}}' \leq {\bm{i}}\}
    $.
    We added for every product gate in $C$ exactly $|A|$ product gates computing $\big(X^{{\bm{i}}'}.f\big)\big(X^{{\bm{i}}-{\bm{i}}'}.g\big)$ for $(\bm{i},\bm{i}') \in A$, 
    and $|A|-1$ addition gates computing $X^{\bm{i}}.(fg) = \sum_{(\bm{i},\bm{i}') \in A}
            \binom{{\bm{i}}}{{\bm{i}}'} 
            \big( X^{{\bm{i}}'}.f \big) \big( X^{{\bm{i}}-{\bm{i}}'}.g \big) 
            $.

    The blow-up of the other steps is bounded above by $|A|$, and the total blow-up is therefore upper bounded by $2|A|-1$.
    We now prove that
$|A| = \binom{L+2K}{2K}$.

Let $(\bm{i},\bm{i}') \in A$. By definition, $i_1 + \cdots + i_K \le L$, and $(i_j-i_j') \ge 0$, for all $j \in [K]$, with $i_j' \ge 0$. Let $a_j:= (i_j -i_j')$, and $b_j:= i_j'$. We want to find the number of $(\bm{a},\bm{b}) \in \IN^{2K}$ such that $\sum_{j=1}^{2K} (a_j+b_j) \le L$. So the number of pairs $(\bm{a},\bm{b}) \in \IN^{2K}$ such that $\sum_{j=1}^{2K} (a_j+b_j) =t$ is $\binom{t+2K-1}{2K-1}$. Therefore, by the hockey-stick identity,
\[
|A|\;=\;\sum_{t=0}^{L} \binom{t+2K-1}{2K-1} \;=\;\binom{L+2K}{2K}.
\]

The above equality readily gives us the upper bound of $|A| = O(L^{2K})$, which gives the desired complexity.
\end{proof}

\begin{proof}[Proof of \Cref{thm:main}]
 Let $\Delta$ be a metapolynomial of format $(d,\delta,n)$ admitting an algebraic circuit of size $s$. For each case of \Cref{thm:main}, we have a corresponding projector:
 \begin{enumerate}
     \item For every weight $\mu$, \Cref{corollary:weight-projector} provides an element $H_\mu \in \calU(\frakgl_k)$ of length at most $O((\delta d)^k)$ acting as projection on the weight space of weight $\mu$. \Cref{proposition:circuit-projection} guarantees that $H_\mu. \Delta$ admits a circuit of size at most $O(s ( (\delta d)^k)^{2k^2}) = O(s (\delta d)^{2k^3})$. 
     \item For every partition $\lambda$, \Cref{corollary:isotypic-component-projector} provides an element $Z_\mu \in \calU(\frakgl_k)$ of length at most $O(k(\delta d)^k)$ acting as projection on the isotypic component of type $\lambda$. \Cref{proposition:circuit-projection} guarantees that $Z_\lambda. \Delta$ admits a circuit of size at most $O(s(k(\delta d)^k)^{2k^2}) =O(sk^{2k^2}(\delta d)^{2k^3})$.
     \item For every partition $\lambda$, \Cref{corollary:highest-weight-projector} provides an element $X_\lambda  \in \calU(\frakgl_k)$ of length at most $O((k+1)(\delta d)^k)$ acting as projection on the highest weight space of weight $\lambda$. \Cref{proposition:circuit-projection} guarantees that $X_\lambda. \Delta$ admits a circuit of size at most $O(s((k+1)(\delta d)^k)^{2k^2}) =O(s (k+1)^{2k^2}(\delta d)^{2k^3})$.
     \item For every semistandard tableaux $T$ of shape $\lambda \partinto d\delta$, \Cref{corollary:semistandard-tableau-projector} provides an element $Y_\lambda\nobreak\in\nobreak\calU(\frakgl_k)$ of length at most $O(k(\delta d)^{k^2})$ acting as projection on $T$-isotypic space.  \Cref{proposition:circuit-projection} guarantees that $Y_T. \Delta$ admits a circuit of size at most $O(s(k(\delta d)^{k^2})^{2k^2}) =O(s k^{2k^2}(\delta d)^{2k^4})$.\qedhere
 \end{enumerate}
\end{proof}

\appendix

\section{Proofs on algebraic natural proofs}
\label{sec:proofsalgnatproofs}

In this section we prove \Cref{thm: nullstellensatz for VP}.

\begin{definition}
    We call a sequence $r \colon \mathbb{N} \to \mathbb{N}$ \emph{strongly superpolynomial} if $r(n)$ dominates every polynomial sequence eventually, that is, for every polynomial $p$ there exists $N$ such that $r(n) > p(n)$ for all $n \geq N$.
\end{definition}
\begin{lemma}
\label{lem:vanishsuperpolynomial}
    $(\Delta_n) \in I(\mathcal C)$ if and only if there exists a strongly superpolynomial sequence $r$ such that $\Delta_n(X_{n, r(n)}) = 0$.
\end{lemma}
\begin{proof}
    Clearly, if $\Delta_n(X_{n, r(n)}) = 0$, then $\Delta_n(X_{n, p(n)})$ vanishes eventually for every polynomially bounded sequence $p$, since $p(n) \leq r(n)$ eventually.
    
    On the other hand, if $(\Delta_n)$ is a sequence of metapolynomials in $I(\mathcal C)$, we can define $r(n)$ as $r(n) := \max \{r \mid \Delta_n(X_{n, r}) = 0\}$.
    For every polynomially bounded sequence $p(n)$ we have that $\Delta_n(X_{n, p(n)})$ vanishes eventually, so $r(n) \geq p(n)$ eventually.
\end{proof}

\begin{proof}[Proof of \Cref{thm: nullstellensatz for VP}]
We first observe we may assume $\deg(f_n) = n$. Indeed, for every p-family $(f_n)$ with the property that $\deg(f_n)$ is strictly increasing, let $(f'_n)$ be the sequence defined by 
\[
f'_m = \left\{ \begin{array}{ll}
f_n & \text{if $m = \deg(f_n)$ for some $n$} \\
0 & \text{otherwise}
\end{array} \right.
\]  
Since $(f_n)$ is a p-family, the sequence $\deg(f_n)$ is polynomially bounded, hence $(f'_n)$ is a p-family as well. Moreover $(f_n) \in \calC$ if and only if $(f_n) \in \calC$, and similarly for $\bar{\calC}$. Finally, for a sequence of metapolynomials $(\Delta_n )$ of format $(*,n,*)$, we have $\Delta_{\deg(f_n)} (f_n) = 0$ eventually if and only if $\Delta_n(f'_n) = 0$ eventually. In particular, proving the statement for $(f'_n)$ is equivalent to proving the statement for $(f_n)$. Since $\deg(f'_n) = n$, this shows we may assume $\deg(f_n) = n$.

Suppose $(f_n) \in \overline{\mathcal C}$, that is, the sequence $p(n) = \underline{c}(f_n)$ is polynomially bounded.
    For every sequence of metapolynomials $(\Delta_n) \in I(\mathcal C)$  we have that $\Delta_n(X_{n, p(n)})$ vanishes eventually.
    Since the metapolynomials $\Delta_n$ are continuous, $\Delta_n(\overline{X_{n, p(n)}})$ vanishes eventually, so $\Delta_n(f_n) = 0$ eventually.

    Conversely, suppose $(f_n) \notin \overline{\mathcal C}$. Therefore $r(n) = \underline{c}(f_n)$ is not polynomially bounded. Let $\theta(n) = \log_n r(n)$, $\hat{\theta}(n) = \max \{\theta(m) \mid m \leq n\}$, and $R(n) = \lceil n^{\hat{\theta}(n)} \rceil$.  Note that $R(n)$ is monotonically increasing.
    
   The sequence $R(n)$ is strongly superpolynomial.
    Indeed, since $r(n)$ is not polynomially bounded, for every $k$ there exists $n_0$ such that $\theta(n_0) \geq k$ and, therefore, $\hat{\theta}(n) \geq k$ for all $n \geq n_0$.
    It follows that for every $k$ we have $R(n) \geq n^k$ eventually.

    For all $\Theta$, if $n_\Theta$ is the minimal value such that $\theta(n_\Theta) \geq \Theta$, then $\hat{\theta}(n_\Theta) = \theta(n_\Theta)$.
    Since $\theta$ is unbounded, it follows that $\theta$ and $\hat{\theta}$ coincide on an infinite subset of $\mathbb{N}$.
    On this subset we also have $R(n) = r(n)$.
    
    Consider a metapolynomial sequence $\Delta_n$ such that $\Delta_n (\overline{X_{n, R(n) - 1}})=0$, and if $r(n) = R(n)$, then $\Delta_n$ is chosen so that $\Delta_n(f_n) \neq 0$ (this is possible because $\underline{c}(f_n) = r(n) = R(n) > R(n) - 1$).
    Since $R(n)$ is strongly superpolynomial, $R(n) - 1$ is strongly superpolynomial, and $(\Delta_n) \in I(\mathcal C)$ by \Cref{lem:vanishsuperpolynomial}, and since $R(n) = r(n)$ infinitely often, $\Delta_n(f_n) \neq 0$ infinitely often.
\end{proof}

\section{PBW, proof that monomials span the space.} \label{sec:pbw}

In this section, we give an elementary proof of the implication in \Cref{thm:pbw} that we use. We include this standard result in this appendix, because it is the main result that enables our algorithmic speedup compared to a naive implementation.

Recall that the universal enveloping algebra $\calU(\frakg)$ is the quotient of the tensor algebra $\calT(\frakg)$ by the ideal $\calI(\frakg) = \{ X \otimes Y - Y \otimes X - [X,Y] : X,Y \in \frakg\}$. The quotient induces a filtration $\calU(\frakg)_{\leq m} \subseteq \calU(\frakg)_{\leq (m+1)}$ and the elements of $\calU(\frakg)_{\leq m}$ are represented by noncommutative polynomials of degree at most $m$ in the elements of $\frakg$. 

\begin{theorem}[One direction of the PBW Theorem]
    \label{theorem:pbw}
    Let $X_1,\ldots,X_r$ be a basis for a Lie algebra $\frakg$. Then  
    \begin{align*}
        \label{equation:pbw-basis}
        \mathcal B_m \coloneqq \big\{ X_{i_1} X_{i_2} \cdots X_{i_\ell}  \mid  i_1 \leq i_2 \leq \cdots \leq i_\ell \text{ and } \ell \leq m\} 
    \end{align*}
    spans $\calU(\frakg)_{\leq m}$.
\end{theorem}
\begin{proof}
It is clear that $\calT(\frakg)_{m}$ is spanned by all elements $X_{i_1} \ootimes X_{i_m}$; hence it suffices to prove that any such element is equivalent, modulo $\calI(\frakg)$, to a linear combination of elements of $\calB_m$. We prove this statement by induction on $m$. 

Let $X_{j_1} \ootimes X_{j_m} \in \calT(\frakg)_{m}$. For every $a = 1 \vvirg m-1$, we have $X_{j_{a}} \otimes X_{j_{a+1}} - X_{j_{a+1}} \otimes X_{j_{a}} - [X_{j_{a}}, X_{j_{a+1}}] \in \calI(\frakg)$; in particular $X_{j_{a}} \otimes X_{j_{a+1}} = X_{j_{a+1}} \otimes X_{j_{a}} + Y$ where $Y$ is a linear combination of elements $\calI(\frakg)$ and of $\calT(\frakg)_{\leq 1}$. We obtain  
    \begin{equation}\label{equation:pbw-transposition-step}
        \begin{split}
        X_{j_1} \ootimes X_{j_\ell}=
        &
       X_{j_1} \ootimes  X_{j_{a-1}} \otimes X_{j_{a+1}} \otimes X_{j_{a}}\otimes X_{j_{a+2}} \ootimes X_{j_\ell} + Y
        \end{split}
    \end{equation}
    where $Y$ is a linear combination of elements of $\calI(\frakg)$ and of $\calT(\frakg)_{\leq m-1}$.

    Let $\pi \in \frakS_m$ be a permutation such that $(j_{\pi(1)} \cdots j_{\pi(m)})$ is nondecreasing.  Write $\pi$ as a composition of elementary transpositions $\pi = \sigma_s \cdots \sigma_1$, where $\sigma_q = (a,a+1)$ for some $a = 1 \vvirg m-1$. Applying \cref{equation:pbw-transposition-step} $s$ times we obtain
\[
        X_{j_1} \ootimes X_{j_m} = X_{j_{\pi(1)}} \ootimes X_{j_{\pi(m)}} + U
 \]
 where $U \in \calI(\frakg) + \calU_{\leq (m-1)}$. Passing to the quotient modulo $\calI(\frakg)$, we obtain the following identity in $\calU(\frakg)_{\leq m}$:
    \[
    X_{j_1}\cdots X_{j_m} = X_{j_{\pi(1)}} \cdots X_{j_{\pi(m)}} +U
    \]
    with $U \in \calU(\frakg)_{\leq (m-1)}$. Apply the induction hypothesis to $U$ to conclude. 
\end{proof}
 
\bibliographystyle{alphaurl}
\bibliography{tensors}

\newcommand{\etalchar}[1]{$^{#1}$}
\begin{thebibliography}{BLMW11}

\bibitem[AC07]{AbCh:BrillGordanLoci}
A.~Abdesselam and J.~Chipalkatti.
\newblock {Brill–Gordan loci, transvectants and an analogue of the Foulkes
  conjecture}.
\newblock {\em Advances in Mathematics}, 208(2):491–520, 2007.
\newblock \href {https://doi.org/10.1016/j.aim.2006.03.003}
  {\path{doi:10.1016/j.aim.2006.03.003}}.

\bibitem[AIR16]{AIR:16}
A.~Abdesselam, C.~Ikenmeyer, and G.~Royle.
\newblock {16,051 formulas for {O}ttaviani's invariant of cubic threefolds}.
\newblock {\em J. Algebra}, 447:649–663, 2016.
\newblock \href {https://doi.org/10.1016/j.jalgebra.2015.11.009}
  {\path{doi:10.1016/j.jalgebra.2015.11.009}}.

\bibitem[AW09]{aaronson2009algebrization}
S.~Aaronson and A.~Wigderson.
\newblock Algebrization: A new barrier in complexity theory.
\newblock {\em ACM Transactions on Computation Theory (TOCT)}, 1(1):1--54,
  2009.
\newblock \href {https://doi.org/10.1145/1490270.1490272}
  {\path{doi:10.1145/1490270.1490272}}.

\bibitem[BB63]{baird-biedenharn-lie-2}
G.~E. Baird and L.~C. Biedenharn.
\newblock On the representations of the semisimple {L}ie groups. {II}.
\newblock {\em J. Mathematical Phys.}, 4:1449--1466, 1963.
\newblock \href {https://doi.org/10.1063/1.1703926}
  {\path{doi:10.1063/1.1703926}}.

\bibitem[BCLR79]{BiCaLoRo79}
D.~Bini, M.~Capovani, G.~Lotti, and F.~Romani.
\newblock {{$O(n\sp{2.7799})$} complexity for {$n\times{}n$} approximate matrix
  multiplication}.
\newblock {\em Inform. Process. Lett.}, 8(5):234–235, 1979.
\newblock \href {https://doi.org/10.1016/0020-0190(79)90113-3}
  {\path{doi:10.1016/0020-0190(79)90113-3}}.

\bibitem[BCS13]{burgisser2013algebraic}
P.~B{\"u}rgisser, M.~Clausen, and M.~A. Shokrollahi.
\newblock {\em Algebraic complexity theory}, volume 315.
\newblock Springer Science \& Business Media, 2013.

\bibitem[BCZ18]{BCZ18}
Markus Bl{\"a}ser, Matthias Christandl, and Jeroen Zuiddam.
\newblock The border support rank of two-by-two matrix multiplication is seven.
\newblock {\em Chicago Journal OF Theoretical Computer Science}, 5:1--16, 2018.

\bibitem[BDI21]{BlDoIk21}
M.~Bläser, J.~Dörfler, and C.~Ikenmeyer.
\newblock {On the Complexity of Evaluating Highest Weight Vectors}.
\newblock {\em 36th Computational Complexity Conference (CCC 2021)},
  200:29:1–29:36, 2021.
\newblock \href {https://doi.org/10.4230/LIPIcs.CCC.2021.29}
  {\path{doi:10.4230/LIPIcs.CCC.2021.29}}.

\bibitem[Ber81]{Berdjis-CasimirOperators}
F.~Berdjis.
\newblock A criterion for completeness of {Casimir} operators.
\newblock {\em J. Math. Phys.}, 22:1851--1856, 1981.
\newblock \href {https://doi.org/10.1063/1.525156}
  {\path{doi:10.1063/1.525156}}.

\bibitem[BGS75]{baker1975relativizations}
T.~Baker, J.~Gill, and R.~Solovay.
\newblock {Relativizations of the $\mathcal{P}=?\mathcal{NP}$ question}.
\newblock {\em SIAM Journal on computing}, 4(4):431--442, 1975.
\newblock \href {https://doi.org/10.1137/020403} {\path{doi:10.1137/020403}}.

\bibitem[BHIM22]{BHIM:22}
P.~Breiding, R.~Hodges, C.~Ikenmeyer, and M.~Michałek.
\newblock {Equations for GL invariant families of polynomials}.
\newblock {\em Vietnam Journal of Mathematics}, 50(2):545–556, 2022.
\newblock \href {https://doi.org/10.1007/s10013-022-00549-4}
  {\path{doi:10.1007/s10013-022-00549-4}}.

\bibitem[BI11]{BI:11}
P.~Bürgisser and C.~Ikenmeyer.
\newblock {Geometric complexity theory and tensor rank}.
\newblock In {\em {Proceedings of the 43rd Annual ACM Symp. on Th. of Comp.}},
  page 509–518, New York, 2011. STOC '11, ACM.
\newblock \href {https://doi.org/10.1145/1993636.1993704}
  {\path{doi:10.1145/1993636.1993704}}.

\bibitem[BI13]{BI:13}
P.~Bürgisser and C.~Ikenmeyer.
\newblock {Explicit Lower Bounds via Geometric Complexity Theory}.
\newblock In {\em {Proceedings of the 45th Annual ACM Symp. on Th. of Comp.}},
  page 141–150. ACM, 2013.
\newblock \href {https://doi.org/10.1145/2488608.2488627}
  {\path{doi:10.1145/2488608.2488627}}.

\bibitem[BI17]{BurIke:FundamentalInvariantsOrbitClosures}
P.~Bürgisser and C.~Ikenmeyer.
\newblock {Fundamental invariants of orbit closures}.
\newblock {\em J. Algebra}, 477:390–434, 2017.
\newblock \href {https://doi.org/10.1016/j.jalgebra.2016.12.035}
  {\path{doi:10.1016/j.jalgebra.2016.12.035}}.

\bibitem[BIP19]{burgisser2019no}
P.~Bürgisser, C.~Ikenmeyer, and G.~Panova.
\newblock {No occurrence obstructions in geometric complexity theory}.
\newblock {\em J. Amer. Math. Soc.}, 32(1):163–193, 2019.
\newblock \href {https://doi.org/10.1090/jams/908}
  {\path{doi:10.1090/jams/908}}.

\bibitem[BLMW11]{BLMW:11}
P.~Bürgisser, J.~M. Landsberg, L.~Manivel, and J.~Weyman.
\newblock {An overview of mathematical issues arising in the {G}eometric
  {C}omplexity {T}heory approach to {$VP \neq VNP$}}.
\newblock {\em SIAM J. Comput.}, 40(4):1179–1209, 2011.
\newblock \href {https://doi.org/10.1137/090765328}
  {\path{doi:10.1137/090765328}}.

\bibitem[Bou05]{Bourbaki-Lie}
N.~Bourbaki.
\newblock {\em Lie groups and {L}ie algebras. {C}hapters 7--9}.
\newblock Elements of Mathematics (Berlin). Springer-Verlag, Berlin, 2005.
\newblock Translated from the 1975 and 1982 French originals by Andrew
  Pressley.

\bibitem[B{\"u}r00]{Bur:00}
P.~B{\"u}rgisser.
\newblock {\em Completeness and reduction in algebraic complexity theory},
  volume~7.
\newblock Springer Science \& Business Media, 2000.
\newblock \href {https://doi.org/10.1007/978-3-662-04179-6}
  {\path{doi:10.1007/978-3-662-04179-6}}.

\bibitem[B{\"u}r24]{burgisser2024completeness}
P.~B{\"u}rgisser.
\newblock Completeness classes in algebraic complexity theory.
\newblock {\em arXiv:2406.06217}, 2024.

\bibitem[Cay45]{Cay:TheoryLinTransformations}
A.~Cayley.
\newblock {On the theory of linear transformations}.
\newblock {\em Cambridge Math. J.}, iv:193–209, 1845.

\bibitem[CIM17]{CIM:17}
M.-W. Cheung, C.~Ikenmeyer, and S.~Mkrtchyan.
\newblock Symmetrizing tableaux and the 5th case of the {F}oulkes conjecture.
\newblock {\em J. Symbolic Comp.}, 80:833--843, 2017.
\newblock \href {https://doi.org/10.1016/j.jsc.2016.09.002}
  {\path{doi:10.1016/j.jsc.2016.09.002}}.

\bibitem[CKR{\etalchar{+}}20]{DBLP:conf/focs/Chatterjee0RST20}
P.~Chatterjee, M.~Kumar, C.~Ramya, R.~Saptharishi, and A.~Tengse.
\newblock {On the Existence of Algebraically Natural Proofs}.
\newblock In {\em {2020 IEEE 61st Annual Symp. on Found. of Comp. Sc. (FOCS)}},
  page 870–880, 2020.
\newblock Full version at \texttt{arXiv:2004.14147}.
\newblock \href {https://doi.org/10.1109/FOCS46700.2020.00085}
  {\path{doi:10.1109/FOCS46700.2020.00085}}.

\bibitem[CKW11]{ChKaWi:PartialDerivativesArithm}
X.~Chen, N.~Kayal, and A.~Wigderson.
\newblock {Partial derivatives in arithmetic complexity and beyond}.
\newblock {\em Found. Trends Theor. Comput. Sci.}, 6(1–2):1–138, 2011.
\newblock \href {https://doi.org/10.1561/0400000043}
  {\path{doi:10.1561/0400000043}}.

\bibitem[DDS21]{DutDwiSax22}
P.~Dutta, P.~Dwivedi, and N.~Saxena.
\newblock {Demystifying the border of depth-3 algebraic circuits}.
\newblock In {\em {2021 IEEE 62nd Annual Symposium on Foundations of Computer
  Science (FOCS)}}, page 92–103. IEEE, 2021.
\newblock \href {https://doi.org/10.1109/FOCS52979.2021.00018}
  {\path{doi:10.1109/FOCS52979.2021.00018}}.

\bibitem[DIP20]{DIP:20}
J.~D{\"o}rfler, C.~Ikenmeyer, and G.~Panova.
\newblock On geometric complexity theory: {M}ultiplicity obstructions are
  stronger than occurrence obstructions.
\newblock {\em SIAM Journal on Applied Algebra and Geometry}, 4(2):354--376,
  2020.
\newblock \href {https://doi.org/10.4230/LIPIcs.ICALP.2019.51}
  {\path{doi:10.4230/LIPIcs.ICALP.2019.51}}.

\bibitem[EGOW18]{EGOW:Barrier}
K.~Efremenko, A.~Garg, R.~Oliveira, and A.~Wigderson.
\newblock {Barriers for Rank Methods in Arithmetic Complexity}.
\newblock In {\em {9th Innovations in Theoretical Computer Science Conference
  (ITCS 2018)}}, volume~94 of {\em {Leibniz International Proceedings in
  Informatics (LIPIcs)}}, page 1:1–1:19, Dagstuhl, Germany, 2018. Schloss
  Dagstuhl – Leibniz-Zentrum für Informatik.
\newblock \href {https://doi.org/10.4230/LIPIcs.ITCS.2018.1}
  {\path{doi:10.4230/LIPIcs.ITCS.2018.1}}.

\bibitem[FH91]{FulHar:RepTh}
W.~Fulton and J.~Harris.
\newblock {\em {Representation theory: a first course}}, volume 129 of {\em
  {Graduate Texts in Mathematics}}.
\newblock Springer-Verlag, New York, 1991.

\bibitem[For14]{forbes2014polynomial}
M.~A. Forbes.
\newblock {\em {Polynomial identity testing of read-once oblivious algebraic
  branching programs}}.
\newblock PhD thesis, Massachusetts Institute of Technology, 2014.

\bibitem[For15]{forbes2015deterministic}
M.~A. Forbes.
\newblock Deterministic divisibility testing via shifted partial derivatives.
\newblock In {\em 2015 IEEE 56th Annual Symposium on Foundations of Computer
  Science}, pages 451--465. IEEE, 2015.
\newblock \href {https://doi.org/10.1109/FOCS.2015.35}
  {\path{doi:10.1109/FOCS.2015.35}}.

\bibitem[For16]{Forbes16}
M.~A. Forbes.
\newblock Some concrete questions on the border complexity of polynomials.
\newblock Presentation given at the Workshop on Algebraic Complexity Theory
  WACT 2016 in Tel Aviv, 2016.
\newblock URL: \url{https://www.youtube.com/watch?v=1HMogQIHT6Q}.

\bibitem[For24]{DBLP:conf/coco/000124}
M.~A. Forbes.
\newblock {Low-Depth Algebraic Circuit Lower Bounds over Any Field}.
\newblock In {\em {39th Computational Complexity Conference, {CCC} 2024, July
  22-25, 2024}}, volume 300 of {\em {LIPIcs}}, page 31:1–31:16. Schloss
  Dagstuhl - Leibniz-Zentrum für Informatik, 2024.
\newblock \href {https://doi.org/10.4230/LIPICS.CCC.2024.31}
  {\path{doi:10.4230/LIPICS.CCC.2024.31}}.

\bibitem[FSV18]{forbes2018succinct}
M.~A. Forbes, A.~Shpilka, and B.~L. Volk.
\newblock Succinct hitting sets and barriers to proving lower bounds for
  algebraic circuits.
\newblock {\em Theory of Computing}, 14(1):1--45, 2018.
\newblock \href {https://doi.org/10.4086/toc.2018.v014a018}
  {\path{doi:10.4086/toc.2018.v014a018}}.

\bibitem[Ga{\l}17]{Gal:VBcactus}
M.~Ga{\l}\k{a}zka.
\newblock {Vector bundles give equations of cactus varieties}.
\newblock {\em Lin. Alg. Appl.}, 521:254–262, 2017.
\newblock \href {https://doi.org/10.1016/j.laa.2016.12.005}
  {\path{doi:10.1016/j.laa.2016.12.005}}.

\bibitem[GGIL22]{GGIL:DegreeRestrictedStrengthDecompsABP}
F.~Gesmundo, P.~Ghosal, C.~Ikenmeyer, and V.~Lysikov.
\newblock {Degree-Restricted Strength Decompositions and Algebraic Branching
  Programs}.
\newblock In {\em {42nd IARCS Ann. Conf. Found. Soft. Tech. and TCS (FSTTCS
  2022)}}, volume 250 of {\em {Leibniz International Proceedings in Informatics
  (LIPIcs)}}, page 20:1–20:15, Dagstuhl, Germany, 2022. Schloss Dagstuhl –
  Leibniz-Zentrum für Informatik.
\newblock \href {https://doi.org/10.4230/LIPIcs.FSTTCS.2022.20}
  {\path{doi:10.4230/LIPIcs.FSTTCS.2022.20}}.

\bibitem[GHL24]{GHL:LinearPreservers}
F.~Gesmundo, Y.-I. Han, and B.~Lovitz.
\newblock {Linear preservers of secant varieties and other varieties of
  tensors}.
\newblock {\em arXiv:2407.16767}, 2024.

\bibitem[GIM{\etalchar{+}}20]{GIMOWW20}
Ankit Garg, Christian Ikenmeyer, Visu Makam, Rafael Oliveira, Michael Walter,
  and Avi Wigderson.
\newblock Search problems in algebraic complexity, {G}{C}{T}, and hardness of
  generators for invariant rings.
\newblock In {\em 35th Computational Complexity Conference}, page~1, 2020.

\bibitem[GIP17]{GesIkPa:GCTMatrixPowering}
F.~Gesmundo, C.~Ikenmeyer, and G.~Panova.
\newblock {Geometric complexity theory and matrix powering}.
\newblock {\em Diff. Geom. Appl.}, 55:106–127, 2017.
\newblock \href {https://doi.org/10.1016/j.difgeo.2017.07.001}
  {\path{doi:10.1016/j.difgeo.2017.07.001}}.

\bibitem[GKKS13]{GKKS:ArithmeticCircuitsChasmDepthThree}
A.~Gupta, P.~Kamath, N.~Kayal, and R.~Saptharishi.
\newblock {Arithmetic circuits: A chasm at depth three}.
\newblock {\em Electronic Colloquium on Computational Complexity (ECCC)},
  20:26, 2013.
\newblock \href {https://doi.org/10.1109/FOCS.2013.68}
  {\path{doi:10.1109/FOCS.2013.68}}.

\bibitem[GKSS17]{grochow2017towards}
J.~A. Grochow, M.~Kumar, M.~Saks, and S.~Saraf.
\newblock Towards an algebraic natural proofs barrier via polynomial identity
  testing.
\newblock {\em arXiv:1701.01717}, 2017.

\bibitem[Gro15]{grochow2015unifying}
J.~A. Grochow.
\newblock Unifying known lower bounds via geometric complexity theory.
\newblock {\em computational complexity}, 24:393--475, 2015.
\newblock \href {https://doi.org/10.1007/s00037-015-0103-x}
  {\path{doi:10.1007/s00037-015-0103-x}}.

\bibitem[GT50]{gelfand-tsetlin}
I.~M. Gelfand and M.~L. Tsetlin.
\newblock Finite-dimensional representations of the group of unimodular
  matrices.
\newblock {\em Dokl. Akad. Nauk SSSR}, 71(8):825, 1950.

\bibitem[Hal10]{Hall:RepTh}
B.~C. Hall.
\newblock {\em {L}ie {G}roups, {L}ie {A}lgebras, and {R}epresentations: {A}n
  {E}lementary {I}ntroduction}, volume 222 of {\em Graduate Texts in
  Mathematics}.
\newblock Springer, New York, 2010.

\bibitem[HIL13]{HIL13}
Jonathan Hauenstein, Christian Ikenmeyer, and Joseph Landsberg.
\newblock Equations for lower bounds on border rank.
\newblock {\em Experimental Mathematics}, 22(4):372--383, 2013.

\bibitem[Hum78]{Hum:LieAlg}
J.~E. Humphreys.
\newblock {\em {Introduction to {L}ie algebras and representation theory}},
  volume~9 of {\em {Graduate Texts in Mathematics}}.
\newblock Springer-Verlag, New York, 1978.

\bibitem[Hum12]{humphreys2012linear}
J.~E. Humphreys.
\newblock {\em {Linear algebraic groups}}, volume~21.
\newblock Springer Science \& Business Media, 2012.

\bibitem[IK20]{IK:20}
C.~Ikenmeyer and U.~Kandasamy.
\newblock {Implementing geometric complexity theory: On the separation of orbit
  closures via symmetries}.
\newblock In {\em {Proceedings of the 52nd Annual ACM SIGACT Symposium on
  Theory of Computing}}, page 713–726, 2020.
\newblock \href {https://doi.org/10.1145/3357713.3384257}
  {\path{doi:10.1145/3357713.3384257}}.

\bibitem[IP17]{IP:17}
C.~Ikenmeyer and G.~Panova.
\newblock {Rectangular Kronecker coefficients and plethysms in geometric
  complexity theory}.
\newblock {\em Adv. Math.}, 319:40–66, 2017.
\newblock \href {https://doi.org/10.1016/j.aim.2017.08.024}
  {\path{doi:10.1016/j.aim.2017.08.024}}.

\bibitem[IS22]{IkSan22}
C.~Ikenmeyer and A.~Sanyal.
\newblock {A note on {VNP}-completeness and border complexity}.
\newblock {\em Information Processing Letters}, 176:106243, 2022.
\newblock \href {https://doi.org/10.1016/j.ipl.2021.106243}
  {\path{doi:10.1016/j.ipl.2021.106243}}.

\bibitem[Kho91]{khovanskiui1991fewnomials}
A.~G. Khovanski\u{i}.
\newblock {\em {Fewnomials}}, volume~88.
\newblock American Mathematical Soc., 1991.

\bibitem[KRST22]{DBLP:conf/stacs/0001RST22}
M.~Kumar, C.~Ramya, R.~Saptharishi, and A.~Tengse.
\newblock {If VNP Is Hard, Then so Are Equations for It}.
\newblock In {\em {39th Int. Symp. on Th. Asp. of Comp. Sc. (STACS 2022)}},
  volume 219 of {\em {Leibniz International Proceedings in Informatics
  (LIPIcs)}}, page 44:1–44:13, Dagstuhl, Germany, 2022. Schloss Dagstuhl –
  Leibniz-Zentrum für Informatik.
\newblock \href {https://doi.org/10.4230/LIPIcs.STACS.2022.44}
  {\path{doi:10.4230/LIPIcs.STACS.2022.44}}.

\bibitem[Kum19]{Kum:QuadraticLowerBound}
M.~Kumar.
\newblock {A quadratic lower bound for homogeneous algebraic branching
  programs}.
\newblock {\em computational complexity}, 28:409–435, 2019.
\newblock \href {https://doi.org/10.1007/s00037-019-00186-3}
  {\path{doi:10.1007/s00037-019-00186-3}}.

\bibitem[KV22]{DBLP:journals/cc/KumarV22}
M.~Kumar and B.~L. Volk.
\newblock A lower bound on determinantal complexity.
\newblock {\em Comput. Complex.}, 31(2):12, 2022.
\newblock \href {https://doi.org/10.1007/S00037-022-00228-3}
  {\path{doi:10.1007/S00037-022-00228-3}}.

\bibitem[LMR13]{LaMaRa:DegDuals}
J.~M. Landsberg, L.~Manivel, and N.~Ressayre.
\newblock {Hypersurfaces with degenerate duals and the {G}eometric {C}omplexity
  {T}heory program}.
\newblock {\em Comment. Math. Helv.}, 88(2):469–484, 2013.
\newblock \href {https://doi.org/10.4171/CMH/292} {\path{doi:10.4171/CMH/292}}.

\bibitem[LO15]{LanOtt:NewLowerBoundsBorderRankMatMult}
J.~M. Landsberg and G.~Ottaviani.
\newblock {New lower bounds for the border rank of matrix multiplication}.
\newblock {\em Th. of Comp.}, 11(11):285–298, 2015.
\newblock \href {https://doi.org/10.4086/toc.2015.v011a011}
  {\path{doi:10.4086/toc.2015.v011a011}}.

\bibitem[LST21]{DBLP:conf/focs/Limaye0T21}
N.~Limaye, Srinivasan S., and S.~Tavenas.
\newblock {Superpolynomial lower bounds against low-depth algebraic circuits}.
\newblock In {\em {2021 IEEE 62nd Annual Symposium on Foundations of Computer
  Science (FOCS)}}, page 804–814. IEEE, 2021.
\newblock \href {https://doi.org/10.1109/FOCS52979.2021.00083}
  {\path{doi:10.1109/FOCS52979.2021.00083}}.

\bibitem[Mac95]{Macdonald-book}
I.~G. Macdonald.
\newblock {\em {Symmetric functions and {H}all polynomials}}.
\newblock {Oxford Mathematical Monographs}. The Clarendon Press Oxford
  University Press, New York, second edition, 1995.

\bibitem[Mah14]{mahajan2014algebraic}
M.~Mahajan.
\newblock Algebraic complexity classes.
\newblock {\em Perspectives in Computational Complexity: The Somenath Biswas
  Anniversary Volume}, pages 51--75, 2014.
\newblock \href {https://doi.org/10.1007/978-3-319-05446-9_4}
  {\path{doi:10.1007/978-3-319-05446-9_4}}.

\bibitem[Mol06]{molev}
A.~I. Molev.
\newblock Gelfand-{T}setlin bases for classical {L}ie algebras.
\newblock In {\em Handbook of algebra. {V}ol. 4}, volume~4 of {\em Handb.
  Algebr.}, pages 109--170. Elsevier/North-Holland, Amsterdam, 2006.
\newblock \href {https://doi.org/10.1016/S1570-7954(06)80006-9}
  {\path{doi:10.1016/S1570-7954(06)80006-9}}.

\bibitem[MS01]{MS01}
K.~D. Mulmuley and M.~Sohoni.
\newblock {Geometric Complexity Theory {I}: An Approach to the {P} vs. {NP} and
  Related Problems}.
\newblock {\em {SIAM} J. Comput.}, 31(2):496--526, 2001.
\newblock \href {https://doi.org/10.1137/S009753970038715X}
  {\path{doi:10.1137/S009753970038715X}}.

\bibitem[MS08]{MS08}
K.~D. Mulmuley and M.~Sohoni.
\newblock Geometric complexity theory {I}{I}: towards explicit obstructions for
  embeddings among class varieties.
\newblock {\em SIAM Journal on Computing}, 38(3):1175--1206, 2008.
\newblock \href {https://doi.org/10.1137/080718115}
  {\path{doi:10.1137/080718115}}.

\bibitem[Mul12]{mulmuley2012gct}
K.~D. Mulmuley.
\newblock {The GCT program toward the P vs. NP problem}.
\newblock {\em Communications of the ACM}, 55(6):98--107, 2012.
\newblock \href {https://doi.org/10.1145/2184319.2184341}
  {\path{doi:10.1145/2184319.2184341}}.

\bibitem[OS16]{OedSam:5thSecant5P1s}
L.~Oeding and S.V Sam.
\newblock {Equations for the fifth secant variety of Segre products of
  projective spaces}.
\newblock {\em Exp. Math.}, 25(1):94–99, 2016.

\bibitem[PP68]{perelomov-popov}
A.~M. Perelomov and V.~S. Popov.
\newblock {Casimir Operators for Semisimple Lie Groups}.
\newblock {\em Mathematics of the USSR-Izvestiya}, 2(6):1313, dec 1968.
\newblock \href {https://doi.org/10.1070/IM1968v002n06ABEH000731}
  {\path{doi:10.1070/IM1968v002n06ABEH000731}}.

\bibitem[Rac50]{racah-rendlincei}
G.~Racah.
\newblock Sulla caratterizzazione delle rappresentazioni irriducibili dei
  gruppi semisemplici di {L}ie.
\newblock {\em Atti Accad. Naz. Lincei Rend. Cl. Sci. Fis. Mat. Nat. (8)},
  8:108--112, 1950.

\bibitem[Reg02]{Reg:02}
K.~W. Regan.
\newblock Understanding the {M}ulmuley--{S}ohoni approach to {P} vs.~{NP}.
\newblock {\em Bulletin of the EATCS}, 78:86--99, 2002.

\bibitem[RR94]{razborov1994natural}
A.~A. Razborov and S.~Rudich.
\newblock Natural proofs.
\newblock In {\em Proc. of the 26th Ann. ACM Symp. Th. of Comp.}, pages
  204--213, 1994.

\bibitem[Sap21]{saptharishi2021survey}
R.~Saptharishi.
\newblock {A survey of lower bounds in arithmetic circuit complexity}.
\newblock \url{https://github.com/dasarpmar/lowerbounds-survey/}, 2021.
\newblock Version 9.0.3.

\bibitem[Syl52]{Sylv:PrinciplesCalculusForms}
J.~J. Sylvester.
\newblock {On the principles of the calculus of forms}.
\newblock {\em Cambridge and Dublin Math. J.}, 7:52–97, 1852.

\bibitem[Val79]{valiant1979completeness}
L.~G. Valiant.
\newblock {Completeness classes in algebra}.
\newblock In {\em {Proceedings of the 11th ACM Symp. on Th. of Comp.}}, {STOC
  '79}, page 249–261, New York, 1979. ACM.
\newblock \href {https://doi.org/10.1145/800135.804419}
  {\path{doi:10.1145/800135.804419}}.

\bibitem[\v{Z}62]{zhelobenko-spectral}
D.~P. \v{Z}helobenko.
\newblock {The classical groups. Spectral analysis of their finite-dimensional
  representations}.
\newblock {\em Uspehi Mat. Nauk}, 17(1(103)):27–120, 1962.
\newblock \href {https://doi.org/10.1070/RM1962v017n01ABEH001123}
  {\path{doi:10.1070/RM1962v017n01ABEH001123}}.

\end{thebibliography}

\end{document}